\documentclass[11pt]{article}
\usepackage[margin=1in]{geometry}
\usepackage{amsmath,amsfonts,amssymb,amsthm}
\usepackage{graphicx}
\usepackage{enumerate}
\usepackage{bbm}
\usepackage{verbatim}
\usepackage{hyperref,color}
\usepackage[capitalize,nameinlink]{cleveref}
\usepackage[dvipsnames]{xcolor}
\hypersetup{
	colorlinks=true,
	pdfpagemode=UseNone,
    citecolor=OliveGreen,
    linkcolor=NavyBlue,
    urlcolor=Magenta,
	pdfstartview=FitW
}
\usepackage{appendix}
\crefname{appsec}{Appendix}{Appendices}
\usepackage{tikz}
\usepackage[ruled]{algorithm2e}
\SetKwComment{Comment}{/* }{ */}

\theoremstyle{plain}
\newtheorem{theorem}{Theorem}[section]

\newtheorem{lemma}[theorem]{Lemma}

\newtheorem{claim}[theorem]{Claim}
\newtheorem{fact}[theorem]{Fact}

\theoremstyle{definition}
\newtheorem{definition}[theorem]{Definition}

\newtheorem{assumption}[theorem]{Assumption}
\newtheorem*{assumption*}{Assumption}

\newtheorem*{condition*}{Condition}

\theoremstyle{remark}
\newtheorem{remark}[theorem]{Remark}

\crefname{lemma}{Lemma}{Lemmas}
\crefname{theorem}{Theorem}{Theorems}
\crefname{definition}{Definition}{Definitions}
\crefname{fact}{Fact}{Facts}
\crefname{claim}{Claim}{Claims}
\crefname{proposition}{Proposition}{Propositions}

\newcommand{\E}{\mathbb{E}}

\newcommand{\one}{\mathbbm{1}}

\DeclareMathOperator*{\argmax}{arg\,max}

\newcommand{\floor}[1]{\left\lfloor #1 \right\rfloor}

\newcommand{\poly}{\mathrm{poly}}

\newcommand{\eps}{\varepsilon}
\renewcommand{\epsilon}{\varepsilon}

\newcommand{\N}{\mathbb{N}}
\newcommand{\Z}{\mathbb{Z}}

\newcommand{\R}{\mathbb{R}}

\newcommand{\GG}{\mathcal{G}}
\newcommand{\XX}{\Omega}

\newcommand{\ccomp}{\mathsf{c}}
\newcommand{\eqand}{\quad\text{and}\quad}

\newcommand{\ham}{h}

\newcommand{\PC}{\mathcal{PC}}
\newcommand{\simt}{\textsc{st}}
\newcommand{\pa}{a}
\newcommand{\qzero}{p}
\newcommand{\rhozero}{\xi}
\newcommand{\size}{s}

\newcommand{\pupp}{\mathcal{P}_{\mathsf{upp}}}
\newcommand{\plow}{\mathcal{P}_{\mathsf{low}}}
\newcommand{\pexp}{\mathcal{P}_{\mathsf{exp}}}
\newcommand{\pgw}{\mathcal{P}_{\mathsf{gw}}}

\usepackage{authblk}

\begin{document}
	
\title{Almost-Linear Planted Cliques \\ Elude the Metropolis Process}
\author[1]{Zongchen Chen \thanks{zongchen@mit.edu}} 
\author[1]{Elchanan Mossel \thanks{elmos@mit.edu}} 
\author[1]{Ilias Zadik \thanks{izadik@mit.edu}}
\affil[1]{Department of Mathematics, MIT}


\date{ \vspace{0.2cm}
\today}

\maketitle

\begin{abstract}
%

A seminal work of Jerrum (1992) showed that large cliques elude the Metropolis process.
More specifically, Jerrum showed that the Metropolis algorithm cannot find a clique of size $k=\Theta(n^{\alpha}), \alpha \in (0,1/2)$, which is planted in the Erd\H{o}s-R\'{e}nyi random graph $G(n,1/2)$, in polynomial time. Information theoretically it is possible to find such planted cliques as soon as $k \ge (2+\eps) \log n$.

Since the work of Jerrum, the computational problem of finding a planted clique in $G(n,1/2)$ was studied extensively and many polynomial time algorithms were shown to find the planted clique if it is of size $k = \Omega(\sqrt{n})$, while no polynomial-time algorithm is known to work when $k=o(\sqrt{n})$. The computational problem of finding a planted clique of $k = o(\sqrt{n})$ is now widely considered as a foundational problem in the study of computational-statistical gaps. Notably, the first evidence of the problem's algorithmic hardness is commonly attributed to the result of Jerrum from 1992.

In this paper we revisit the original Metropolis algorithm suggested by Jerrum.
Interestingly, we find that the Metropolis algorithm actually {\em fails} to recover a planted clique of size $k=\Theta(n^{\alpha})$ for {\em any} constant $0 \leq \alpha < 1$, 
unlike many other efficient algorithms that succeed when $\alpha>1/2$. Moreover, we strengthen Jerrum's results in a number of other ways including:

\begin{itemize}

\item Like many results in the MCMC literature, the result of Jerrum shows that there exists a starting state (which may depend on the instance) for which the Metropolis algorithm fails to find the planted clique in polynomial time. For a wide range of temperatures, we show that the algorithm fails when started at the most natural initial state, which is the empty clique. This answers an open problem stated in Jerrum (1992). It is rather rare to be able to show the failure of a Markov chain starting from a specific state.

\item We show that the simulated tempering version of the Metropolis algorithm, a more sophisticated temperature-exchange variant of it, also fails at the same regime of parameters.

\end{itemize}

Our results substantially extend Jerrum's result. Furthermore, they confirm recent predictions by Gamarnik and Zadik (2019) and Angelini, Fachin, de Feo (2021). Finally, they highlight the subtleties of using the sole failure of one, however natural, family of algorithms as a strong sign of a fundamental statistical-computational gap.
\end{abstract}

\thispagestyle{empty}

\newpage
\tableofcontents
\thispagestyle{empty}

\newpage
\setcounter{page}{1}
\section{Introduction}

The problem of finding in polynomial-time large cliques in the $n$-vertex Erd\H{o}s-R\'{e}nyi random graph $\GG(n,1/2),$ where each edge is present independently with probability $1/2$, is a fundamental open problem in algorithmic random graph theory \cite{karp1979recent}. In $\GG(n,1/2)$ it is known that there is a clique of size $(2-o(1))\log n$ with high probability (w.h.p.) as $n \rightarrow +\infty$, and several simple polynomial-time algorithms can find a clique of size $(1+o(1))\log n$ w.h.p.  which is nearly half the size of the maximum clique (e.g. see \cite{grimmett_mcdiarmid_1975}). (Note that here and everywhere we denote by $\log$ the logarithm with base 2.) Interestingly, there is no known polynomial-time algorithm which is able to find a clique of size $(1+\epsilon)\log n$ for some constant $\epsilon>0$. The problem of finding such a clique in polynomial-time was suggested by Karp \cite{karp1979recent} and still remains open to this day.

\paragraph{Jerrum's result and the planted clique model} Motivated by the challenge of finding a clique in $\GG(n,1/2)$, Jerrum in \cite{Jer92} established that \emph{large cliques elude the Metropolis process} in $\GG(n,1/2)$. Specifically, he considered the following {\em Gibbs Measure} for $\beta> 0,$ $G \sim \GG(n,1/2)$
\begin{equation}
\label{Gibbs_intro} 
\pi_{\beta}(C) \propto \exp( \beta |C|),
\end{equation} 
where $C \subseteq V(G)$ induces a clique in the instance $G$. Notice that since $\beta>0,$ by definition $\pi_{\beta}$ assigns higher mass on cliques of larger size. Jerrum considered the Metropolis process with stationary measure $\pi_{\beta}$, which is initialized in some clique, say $X_0$ of $G.$ Then the process moves ``locally" between cliques which differ in exactly one vertex. In more detail, every step of the Metropolis process is described as follows (see also Algorithm \ref{alg:metropolis_intro}). 
Choose a vertex $v \in V$ uniformly at random. If $v \in X_t$ where $X_t$ is the current clique, then let $X_{t+1} = X_t \setminus \{v\}$ (a ``downward" step) with probability $e^{-\beta}$ and let $X_{t+1} = X_t$ with the remaining probability; else if $v \notin C_t$ and $X_t \cup \{v\}$ is a clique, then let $X_{t+1} = X_t \cup \{v\}$ (an ``upward" step);  otherwise, let $X_{t+1} = X_t$.  For this process, Jerrum established the negative result that it fails to reach a clique of size $(1+\epsilon)\log n$ in polynomial-time for any constant $\epsilon>0$. We note that, as customary in the theory of Markov chains, the failure is subject to the Metropolis process being ``worst-case initialized", that is starting from some ``bad" clique $X_0$. This is a point we revisit later in this work.

The {\em planted clique problem} was introduced by Jerrum in \cite{Jer92} in order to highlight the failure of the Metropolis process. For $k \in \mathbb{N}, k \leq n$ the \emph{planted clique model} $\GG(n,1/2,k)$ is defined 
by first sampling an instance of $\GG(n,1/2)$, then choosing $k$ out of the $n$ vertices uniformly at random and finally adding all the edges between them (if they did not already exist from the randomness of $\GG(n,1/2)$). The set of $k$ chosen vertices is called the planted clique $\mathcal{PC}$. It is perhaps natural to expect that the existence of $\mathcal{PC}$ in $G$ can assist the Metropolis process to reach faster cliques of larger size. Yet, Jerrum proved that as long as $k=\floor{n^{\alpha}}$ for some constant $\alpha<1/2$ the Metropolis process continues to fail to find a clique of size $(1+\epsilon)\log n$ in polynomial-time, for any $\epsilon>0$. As he also noticed when $\alpha>1/2$ one can actually trivially recover $\mathcal{PC}$ from $G$ via a simple heuristic which chooses the top-$k$ degrees of the observed graph (see also \cite{kuvcera1995expected}). In particular, one can trivially find a clique of size much larger than $\log n$ when $\alpha>1/2.$ Importantly though, he \emph{never proved} that the Metropolis process actually succeeds in finding large cliques when $\alpha>1/2,$ leaving open a potentially important piece of the performance of the Metropolis process in the planted clique model. In his words from the conclusion of \cite{Jer92}:

\begin{quotation}
\textit{``If the failure of the Metropolis process to reach $(1+\epsilon)\log n$-size cliques could be shown to hold for some $\alpha>1/2$, it would represent a severe indictment of the Metropolis process as a heuristic search technique for large cliques in a random graph.''}
\end{quotation} In this work, we seek to investigate the performance of the Metropolis process for all $\alpha \in (0,1).$

\begin{algorithm}
\caption{Metropolis Process \cite{Jer92}}\label{alg:metropolis_intro}
\KwIn{a graph $G$, a starting clique $X_0 \in \XX$, stopping time $T$}
\For{$t = 1,\dots,T$}{
  Pick $v\in V$ uniformly at random\;
  $C \gets X_{t-1} \oplus \{v\}$\;
  \eIf{$C \in \XX$}{
    $X_t \gets 
    \begin{cases}
    C, &\text{with probability~} \min\{1,\exp(\beta(|C|-|X_t|))\};\\
    X_{t-1}, &\text{with remaining probability};
    \end{cases}$
  }{$X_t \gets X_{t-1}$\;
  }
}
\KwOut{$X_T$}
\end{algorithm}

\paragraph{The planted clique conjecture} 
Following the work of Jerrum \cite{Jer92} and Kucera \cite{kuvcera1995expected} the planted clique model 
has received a great deal of attention 
and became a hallmark of a research area that is now called {\em study of statistical-computational gaps}. 
The planted clique problem can be phrased as a statistical or inference problem in the following way: given an instance of $\GG(n, 1/2, k)$ recover $\mathcal{PC},$ the planted clique vertices. It is impossible to recover the clique when it is of size $k<(2-\epsilon)\log n$ for any constant $\epsilon>0$ (see e.g. \cite{arias2014community}), but possible information theoretically (and in quasi-polynomial time $n^{O(\log n)}$) whenever $k > (2+\eps) \log n$ for any constant $\eps > 0$ (see e.g. the discussion in \cite{feldman2017statistical}). Meanwhile, multiple polynomial-time algorithms have been proven to succeed in recovering the planted clique but only under the much larger size $k=\Omega(\sqrt{n})$ \cite{alon1998finding, ron2010finding, dekel2014finding, deshpande2015finding}. The intense study of the model, as well as the failure to find better algorithms, has lead to \emph{the planted clique conjecture}, stating that the planted clique recovery task is impossible in polynomial-time, albeit information-theoretically possible, whenever $(2+\epsilon) \log n \le k=o(\sqrt{n})$. The planted clique conjecture has since lead to many applications, a highlight of which is that it serves as the main starting point for building reductions between a plethora of statistical tasks and their computational-statistical trade-offs (see e.g. \cite{berthet2013complexity, ma2015computational, brennan2020reducibility}).

Unfortunately, because of the average-case nature of the planted clique model, a complexity theory explanation is still lacking for the planted clique conjecture. For this reason, researchers have so far mainly focused on supporting the conjecture by establishing the failure of restricted families of polynomial-time methods, examples of which are Sum-of-Squares lower bounds \cite{barak2019nearly}, statistical-query lower bounds \cite{feldman2017statistical} and low-temperature MCMC lower bounds \cite{gamarnik2019landscape}. Because of the focus on establishing such restricted lower bounds, the vast majority of works studying the planted clique conjecture cite the result of Jerrum \cite{Jer92} on the failure of the Metropolis process as the ``first evidence" for the validity of it. Note though that given what we discussed above, such a claim for Jerrum's result can be problematic. Indeed, recall that Jerrum have not established the success of the Metropolis process when $k=\floor{n^{\alpha}}, \alpha>1/2$ in reaching cliques of size $(1+\epsilon)\log n$, let alone identifying the planted clique. That means that the Metropolis process could in principle simply be not recovering the planted clique for all $\alpha \in (0,1),$ offering no evidence for the planted clique conjecture.


In fact, in 2019, Gamarnik and Zadik \cite{gamarnik2019landscape} studied the performance of various MCMC methods (not including the Metropolis process described above) for the planted clique model and conjectured their failure to recover the planted clique when $k<n^{2/3}$, that is the failure much beyond the $\sqrt{n}$ threshold. For this reason, they raised again the question whether Jerrum's original Metropolis process actually fails for $k$ larger than $\sqrt{n}$. Later work by Angelini, Fachin and de Feo \cite{angelini2021mismatching} simulated the performance of the Metropolis process and predicted its failure to recover the planted clique again much beyond the $\sqrt{n}$ threshold. Both of these results suggest that the Metropolis process may not be able to recover the planted clique for some values of $k=n^{\alpha}, \alpha>1/2$. We consider the absence of a negative or positive result for whether the Metropolis process of \cite{Jer92} recovers the planted clique a major gap in the literature of the planted clique conjecture which we investigate in this work.

\paragraph{Empty clique initialization}

A common deficiency of many Markov chain lower bounds is their failure to establish lower bounds for starting from any particular state. Indeed, many of the lower bounds in the theory of Markov chains including spectral and conductance lower bounds are proved with high probability or in expectation over the stationary measure of the chain. Of course, since the lower bounds provide usually evidence that is hard to sample from the stationary distribution, the deficiency of the lower bound is that it is hard to find a state where the lower bound applies! 
This is indeed the case for the specific example of the Metropolis process in \cite{Jer92}, where a-priori, other than the empty set we do not know which subsets of the nodes make a clique. 

Jerrum noted this deficiency in his paper~\cite{Jer92}:
\begin{quotation}
\textit{``The most obviously unsatisfactory feature of Theorems ... is that these theorems assert only the existence of a starting state from which the Metropolis process takes super-polynomial time to reach a large clique ... It seems almost certain that the empty clique is a particular example of a bad starting state, but different proof techniques would be required to demonstrate this.''}
\end{quotation}
In this work we develop a new approach based on comparison with birth and death processes allowing us to prove lower bounds starting from the empty clique. We note that previous work on birth and death processes established Markov chain bounds starting from a specific state~\cite{diaconis2006separation}.

\section{Main Contribution}

We now present our main contributions on the performance of the Metropolis process for the planted clique model $\GG(n,1/2,k),$ where $k=\floor{n^{\alpha}}$ for some constant $\alpha \in (0,1)$. Our results hold for the Metropolis process associated with any Gibbs measure defined as follows. For an arbitrary Hamiltonian vector $h=(h_q, q \in [n])$ and arbitrary $\beta \geq 0,$ let
\begin{align}\label{gibbs_contrib}
    \pi_{\beta}(C) \propto \exp(\beta h_{|C|}), C \in \XX(G)
    \end{align} where by $\XX(G)$ we denote the cliques of $G$. In some results, we would require a small degree of regularity from the vector $h$, which for us is to satisfy $h_0=0$ and to be 1-Lipschitz in the sense that \begin{align}
        |h_q-h_{q'}| \leq |q-q'|, \forall q,q' \in [\floor{2\log n}].
    \end{align}We call these two conditions, as simply $h$ being ``regular". The regularity property is for technical reasons, as it allows to appropriately bound the transition probabilities of the Metropolis process between small cliques. Notice that $h_q=q, q\in [n]$ is trivially regular and corresponds to the Gibbs measure and Metropolis process considered by Jerrum, per \eqref{Gibbs_intro}.

Our first theorem is a very general lower bound which holds under worst-case initialization for all $\alpha \in (0,1).$

\begin{theorem}[Informal version of Theorem \ref{main_result_mixing} and Theorem \ref{thm:large-clique}]\label{main_result_mixing_contrib}
Let $k=\floor{n^{\alpha}}$ for any $\alpha \in (0,1)$.
\begin{enumerate}
    \item[I.] For arbitrary $\ham$ and arbitrary inverse temperature $\beta$, for any constant $\gamma>0$ there exists a ``bad" initialization such that  the Metropolis process requires $n^{\Omega(\log n)}$ time to reach $\gamma \log n$ intersection with the planted clique.
    
    \item[II.] For arbitrary regular $\ham$ and arbitrary inverse temperature $\beta=O(\log n)$, for any constant $\epsilon>0,$ there exists a ``bad" initialization such that the Metropolis process requires $n^{\Omega(\log n)}$ time to reach a clique of size $(1+\epsilon)\log n$.
\end{enumerate}
\end{theorem}


One way to think about the two parts of \cref{main_result_mixing_contrib} is in terms of the statistical failure and the optimization failure of the Metropolis algorithm. The first part establishes the statistical failure as the algorithm cannot even find $\gamma \log n$ vertices of the planted clique. The second part shows that it fails as an optimization algorithm: the existence of a huge planted clique still does not 
improve the performance over the $(1+\eps) \log n$ level. 


Note that second part extends the result of \cite{Jer92} to all $\alpha \in (0,1),$ when $\beta=O(\log n)$. (The case $\beta=\omega(\log n)$ is proven below in \cref{thm:greedy_contr} since for that low temperature the process behaves like the greedy algorithm). In Jerrum's words, our result reveals ``a severe indictment of the Metropolis process in finding large cliques in random graphs"; even a $n^{1-\delta}$-size planted clique does not help the Metropolis process to reach cliques of size $(1+\epsilon)\log n$ in polynomial time.  At a technical level it is an appropriate combination of the bottleneck argument used by Jerrum in \cite{Jer92}, which focus on comparing cliques based on their size, and a separate bottleneck argument comparing cliques based on how much they intersect the planted clique.

Our next result concerns the case where the Metropolis process is initialized from the empty clique. We obtain for all $\beta=o(\log n)$ the failure of the Metropolis process starting from any $o(\log n)$-size clique (including in particular the empty clique). In particular, this answers in the affirmative the question from \cite{Jer92} for all $\beta=o(\log n).$ 

\begin{theorem}[Informal version of Theorem \ref{thm:start-from-empty} and Theorem \ref{thm:start-from-empty-hit-large}]\label{thm:start-from-empty_contr}

Let $k=\floor{n^{\alpha}}$ for any $\alpha \in (0,1)$.
Then for arbitrary regular $\ham$ and arbitrary inverse temperature $\beta=o(\log n)$, the Metropolis process started from any clique of size $o(\log n)$ (in particular the empty clique) requires $n^{\Omega(\log n)}$ time to reach a clique which for some constants $\gamma, \epsilon>0$ either
\begin{itemize}
    \item has at least $\gamma \log n$ intersection with the planted clique or,
    \item has size at least $(1+\epsilon)\log n.$
\end{itemize}
\end{theorem}
The proof of \cref{thm:start-from-empty_contr} is based on the expansion properties of all $(1-\epsilon)\log n$-cliques of $G(n,1/2)$. 
The expansion properties allow us to compare the Metropolis process to an one dimensional birth and death process that keeps track of the size of the clique (or the size of the intersection with the hidden clique). The analysis of this process is based on a time-reversal argument. 

One can wonder, whether a lower bound can be also established in the case $\beta=\Omega(\log n)$ when we start from the empty clique. We partially answer this, by obtaining the failure of the Metropolis process starting from the empty clique in the case where $h_q=q, q \in [n]$ and $\beta=\omega(\log n)$. 
\begin{theorem}[Informal version of \cref{thm:greedy}]\label{thm:greedy_contr}

Let $k=\floor{n^{\alpha}}$ for any $\alpha \in (0,1)$.
For $\ham_q=q, q \in [n]$ and arbitrary inverse temperature $\beta=\omega(\log n)$, the Metropolis process started from any clique of size $o(\log n)$ (in particular the empty clique) requires $n^{\omega(1)}$ time to reach a clique which for some constants $\gamma, \epsilon>0$ either
\begin{itemize}
    \item has at least $\gamma \log n$ intersection with the planted clique or,
    \item has size at least $(1+\epsilon)\log n.$
\end{itemize}
\end{theorem}

The key idea here is that for $h_q=q,$ if $\beta=\omega(\log n)$ then with high probability the Metropolis process
never removes vertices, so it is in fact the same as the greedy algorithm. 
This observation allows for a much easier analysis of the algorithm. 

In a different direction, Jerrum in \cite{Jer92} asked whether one can extend the failure of the Metropolis process to the failure of simulated annealing on finding large cliques in random graphs. We make a step also in this direction, by considering the simulated tempering (ST) version of the Metropolis process \cite{marinari1992simulated}. The simulated tempering is a celebrated Monte Carlo scheme originated in the physics literature that considers a Gibbs measure, say the one in \eqref{gibbs_contrib}, at different temperatures $\beta_1<\beta_2<\ldots<\beta_m$, in other words considers the Gibbs measures $\pi_{\beta_1},\pi_{\beta_2},\ldots, \pi_{\beta_M}$. Then it runs a Metropolis process on the product space between the set of temperatures and the Gibbs measures, which allows to modify the temperature during its evolution and interpolate between the different $\pi_{\beta_i}$ (for the exact definitions see Section \ref{sec:st_dfn}). The ST dynamics have been observed extensively in practise to outperform their single temperature Metropolis process counterparts but rigorous analysis of these processes is rather rare, see~\cite{bhatnagar2004torpid}.



It turns out that our ``bad" initialization results extend in a straightforward manner to the ST dynamics.

\begin{theorem}[Informal version of Theorem \ref{main_result_mixing_st}, and Theorem \ref{thm:large-clique_st}]\label{st_bad_contr}
Let $k=\floor{n^{\alpha}}$ for any $\alpha \in (0,1)$.
\begin{itemize}
    \item For arbitrary $\ham$, arbitrary $m \in \N$ and arbitrary sequence of inverse temperatures $\beta_1<\beta_2<\ldots<\beta_m$, there exists a ``bad" initialization such as the ST dynamics require $n^{\Omega(\log n)}$ time to reach $\epsilon \log n$ intersection with the planted clique, for some constant $\epsilon>0.$
    
    \item For arbitrary regular $\ham$, arbitrary $m \in \mathbb{N}$ and arbitrary sequence of inverse temperatures $\beta_1<\beta_2<\ldots<\beta_m$ with $\max_i |\beta_i|=O(\log n)$, there exists a ``bad" initialization such as the ST dynamics require $n^{\Omega(\log n)}$ time to reach a clique of size $(1+\epsilon)\log n$-size for some constant $\epsilon>0.$
\end{itemize}
\end{theorem}The key idea behind the proof of \cref{st_bad_contr} is that the bottleneck set considered in the proof of Theorem \ref{main_result_mixing_contrib} is ``universal", in the sense that the same bottleneck set applies to all inverse temperatures $\beta.$ For this reason, the same bottleneck argument can be applied to simulated tempering that is allowed to change temperatures during its evolution.

On top of this, we establish a lower bound on the ST dynamics, when started from the empty clique.

\begin{theorem}[Informal version of \cref{thm:Simulated-Tempering-empty}]\label{thm:st_empty}

Let $k=\floor{n^{\alpha}}$ for any $\alpha \in (0,1)$.
For arbitrary regular and increasing $\ham_q=q, q \in [n],$ arbitrary $m=o(\log n)$ and arbitrary sequence of inverse temperatures $\beta_1<\beta_2<\ldots<\beta_m$ with $\max_i |\beta_i|=O(1)$, the ST dynamics started from the pair of the empty clique and the temperature $\beta_1$ require $n^{\omega(1)}$ time to reach a clique which for some constant $\gamma, \epsilon>0$ either
\begin{itemize}
    \item has at least $\gamma \log n$ intersection with the planted clique or,
    \item has size at least $(1+\epsilon)\log n.$
\end{itemize}
\end{theorem}

\subsection{Further Comparison with Related Work}

\paragraph{Comparison with \cite{angelini2021mismatching}} As mentioned in the Introduction, the authors of \cite{angelini2021mismatching} predicted the failure of the Metropolis process for the Gibbs measure \eqref{Gibbs_intro} to recover the planted clique. Specifically, based on simulations they predicted the failure for all $k=\floor{n^{\alpha}},$ for $\alpha < \alpha^* \approx 0.91,$ though they comment that the exact predicted threshold $\alpha^*\approx 0.91$ could be an artifact of finite sized effects. In this work, using \cref{main_result_mixing_contrib} we establish that (worst-case initialized) the Metropolis process fails for all $\alpha \in (0,1),$ confirming their prediction but for $\alpha^*=1.$ In the same work, the authors suggest studying instead the Metropolis process for another Gibbs measure of the form \eqref{gibbs_contrib}, which they call BayesMC. Their suggested Gibbs measure for a specific value of $\beta$ matches a (slightly perturbed) version of the posterior of the planted clique recovery problem. The authors predict based on simulations that by appropriately turning (and ``mismatching") $\beta>0$, the Metropolis chain now recovers the planted clique for all $k=\floor{n^{\alpha}}, \alpha>1/2$.  In this work, we partially refute this prediction using \cref{main_result_mixing_contrib}, as (worst-case initialized) the Metropolis process for any Gibbs measure \eqref{gibbs_contrib}, including the mismatched posterior that \cite{angelini2021mismatching} considers, fails to recover the planted clique for all $\alpha \in (0,1).$ 

\paragraph{MCMC underperformance in statistical inference}
Our lower bounds show the suboptimality of certain natural MCMC methods in inferring the planted clique in a regime where simple algorithms work. Interestingly, this agrees with a line of work establishing the suboptimality of MCMC methods in inferring hidden structures in noisy environments. Such a phenomenon have been generally well-understood in the context of tensor principal component analysis \cite{richard2014statistical, benarousPCA}, where Langevin dynamics and gradient descent on the empirical landscape are known to fail to infer the hidden tensor, when simple spectral methods succeed.  Moreover, this suboptimality has been recently observed through statistical physics methods for other models including sparse PCA \cite{Antenucci19}, the spiked matrix-tensor model \cite{pmlr-v97-mannelli19a} and phase retrieval \cite{mannelli2020complex}. Finally, it has been also observed for a different family of MCMC methods but again for the planted clique model in \cite{gamarnik2019landscape}.

\section{Proof Techniques and Intuitions}

In this section, we offer intuition regarding the proofs of our results.  In the first two subsections, we make a proof overview subsection we provide intuition behind our worst-case initialization results for the Metropolis process \cref{main_result_mixing_contrib} and ST dynamics \cref{st_bad_contr}. In the following one we discuss about our results on the Metropolis process and ST dynamics with the empty clique initialization \cref{thm:start-from-empty,thm:start-from-empty-hit-large}.

\subsection{Worst-case Initialization for Reaching Large Intersection}

We start with discussing the lower bound for obtaining $\gamma \log n$ intersection with the planted clique. We employ a relatively simple bottleneck argument, based on \cref{lem:conductance}.

For $q,r \in \mathbb{N}$ let us denote by $W_{q,r}$ the number of cliques in $G \sim \GG(n,1/2,k)$ which have size $q$ and intersect the planted clique exactly on $r$ vertices. Our first starting point towards building the bottleneck is the following simple observation. For any $q<2\log n,$ and $r=\gamma \log n$ for a constant $0<\gamma<2(1-\alpha)$ we have w.h.p. as $n \rightarrow +\infty.$
\begin{align}\label{ratio}
W_{q,r}/W_{q,0} \leq n^{-\Omega(\log n)}.
\end{align}In words, the number of $q$-cliques of intersection $r$ with the planted clique are a quasi-polynomial factor less than the number of $q$-cliques which are disjoint with the planted clique. 
Indeed, at the first-moment level it can be checked to hold $\E[W_{q,r}]/\E[W_{q,0}]=\exp((-((1-\alpha)\gamma+\frac{\gamma^2}{2})\ln 2(\log n)^2)$ and a standard second moment argument gives the result (both results are direct outcomes of \cref{basic_rg}). Notice that this property holds for all $\alpha \in (0,1).$

Now let us assume for start that $\beta=0$. In that case the Gibbs measure $\pi_0$ is simply the uniform measure over cliques of $G$. For $r=\gamma \log n$ let $\mathcal{A}_r$ be the subset of the cliques with intersection with the planted clique at most $r$. Using standard results (see \cref{lem:conductance}) to prove our hitting time lower bound it suffice to show for any $r=\gamma \log n$ with constant $\gamma>0,$
\begin{align}\label{eq:clique_bound}
   \frac{\pi_0(\partial \mathcal{A}_r)}{\pi_0(\mathcal{A}_r)}= \frac{\sum_{q=1}^n W_{q,r}}{\sum_{q \in [n], 0 \leq s \leq r} W_{q,s}} \leq n^{-\Omega(\log n)},
\end{align}w.h.p. as $n \rightarrow +\infty$. Indeed, given such result based on \cref{lem:conductance} there is an initialization of the Metropolis process in $\mathcal{A}_r$ from which it will take quasi-polynomial time to reach the boundary of $\mathcal{A}_r$, which are exactly the cliques of intersection $r$ with the planted clique.

Now another first moment argument, (an outcome of part (3) of \cref{basic_rg}) allows us to conclude that since $\gamma<2(1-\alpha)$ is not too large, the only cliques of intersection $r=\gamma \log n$ with the planted clique satisfy $q<(2-\epsilon)\log n$ for small enough $\epsilon>0.$ In other words $W_{q,r}=0$ unless $q<(2-\epsilon)\log n.$
But now using also \eqref{ratio} we have 
\begin{align*}
    \sum_{q \in [n], 0 \leq s \leq r} W_{q,s} \geq \sum_{q \in [(2-\epsilon)\log n]} W_{q,0} \geq n^{\Omega(\log n)}\sum_{q \in [(2-\epsilon)\log n]} W_{q,r}=n^{\Omega(\log n)}\sum_{q \in [n]} W_{q,r},
\end{align*}and the result follows. 

Now, the interesting thing is that the exact same calculation works for arbitrary $\beta \geq 0$ and arbitrary Hamiltonian vector $h.$ Consider as before the same subset of cliques $\mathcal{A}_r.$ It suffices to show
\begin{align}\label{gen_beta}
   \frac{\pi_{\beta}(\partial \mathcal{A}_r)}{\pi_{\beta}(\mathcal{A}_r)}= \frac{\sum_{q=1}^n W_{q,r}e^{\beta h_q}}{\sum_{q \in [n], 0 \leq s \leq r} W_{q,s}e^{\beta h_q}} \leq n^{-\Omega(\log n)},
\end{align}But as before
\begin{align*}
    \sum_{q \in [n], 0 \leq s \leq r} W_{q,s}e^{\beta h_q} \geq \sum_{q \in [(2-\epsilon)\log n]} W_{q,0}e^{\beta h_q} \geq n^{\Omega(\log n)}\sum_{q \in [(2-\epsilon)\log n]} W_{q,r}e^{\beta h_q}=n^{\Omega(\log n)}\sum_{q \in [n]} W_{q,r}e^{\beta h_q},
\end{align*}and the general result follows.

Interestingly, this bottleneck argument transfers easily to a bottleneck for the ST dynamics. The key observation is that to construct the bottleneck set for the Metropolis process, we used \emph{the same bottleneck set} $\mathcal{A}_r$ for all inverse temperatures $\beta$ and Hamiltonian vectors $h$. Now the ST dynamics are a Metropolis process on the enlarged product space of the temperatures times the cliques. In particular, the stationary distribution of the ST dynamics is simply a mixture of the Gibbs distributions $\pi_{\beta_i}, i \in [m]$ say $\pi_{\mathrm{ST}}=\sum_{i=1}^m v_i \pi_{\beta_i}$ for some weights $v_i, i \in [m]$ (see Section \ref{sec:st_dfn} for the exact choice of the weights $v_i, i \in [m]$, though they are not essential for the argument). Now using \eqref{gen_beta} and that the boundary operator satisfies $\partial ([m]\times\mathcal{A}_r)= [m]\times\partial(\mathcal{A}_r)$ we immediately have
\begin{align}\label{gen_beta_2}
  \frac{\pi_{\mathrm{ST}}(\partial ([m]\times\mathcal{A}_r))}{\pi_{\mathrm{ST}}([m]\times \mathcal{A}_r)}= \frac{\sum_{i=1}^m v_i \pi_{\beta_i}(\partial \mathcal{A}_r)}{\sum_{i=1}^m v_i \pi_{\beta_i}(\mathcal{A}_r)}\leq n^{-\Omega(\log n)}.
\end{align}The result then follows again using Lemma \ref{lem:conductance}.

\subsection{Worst-case Initialization for Reaching Large Cliques}

 The worst-case initialization lower bound for reaching $(1+\epsilon)\log n$ cliques is also based on a bottleneck argument. This time the argument is more involved and is a combination of the bottleneck argument by Jerrum in \cite{Jer92} which worked for $\alpha<1/2$ and the bottleneck argument used in the previous subsection for reaching large intersection with the planted clique which worked for all $\alpha<1.$

We first need some notation. For $r\leq q$ we denote by $\XX_{q,r}$ the set of $q$-cliques of intersection $r$ with the planted clique and $\XX_{q,*}$ the set of all cliques of size $q$.  We also define $\XX_{q,<r}$ the set of $q$-cliques of intersection less than $r$ with the planted clique and analogously by $\XX_{<q,r}$ the set of cliques of size less than $q$ and intersection $r$ with the planted clique

We first quickly remind the reader the argument of Jerrum. Recall the notion of a $q$-gateway clique from \cite{Jer92}, which informally is the last clique of its size in some path starting from the empty clique and ending to a clique of size $q$ (see Section \ref{sec:bad_large} for the exact definition). For $q$ we call $\Psi_{q}$ the set of of $q$-gateway cliques. Importantly, for any $p<q$ any path from the empty clique to a $q$-clique crosses from some $q$-gateway of size $p.$  Jerrum's bottleneck argument in \cite{Jer92} is then based on the observation that assuming $\alpha<1/2$ for $\epsilon>0$ a small enough constant, if $p=\floor{(1+2\epsilon/3)\log n}$ and $q=\floor{(1+\epsilon)\log n}$ then \begin{align}\label{eq:gate_jer}|\Psi_{q} \cap \XX_{p,*}|/|\XX_{p, *}| \leq n^{-\Omega(\log n)}.
\end{align}Unfortunately, such a bottleneck argument is hopeless if $\alpha>1/2$ since in that case most cliques of size at most $q$ are fully included in the planted clique and the ratio trivializes.

We identify the new bottleneck by leveraging more the ``intersection axis" with the planted clique". Our first observation is that a relation like \eqref{eq:gate_jer} holds actually for all $\alpha<1$ if the cliques are restricted to have low-intersection with the planted clique, i.e. for all $\alpha \in (0,1)$ if $r=\gamma \log n$ for small enough constant $\gamma>0$ and $p,q$ defined as above then it holds
\begin{align}\label{cond_large_1}
    |\Psi_{q} \cap \XX_{p,<r}|/|\XX_{p, <r}| \leq n^{-\Omega(\log n)}.
\end{align} The second observation, is that for the Metropolis process to hit a clique of size $q=(1+\epsilon)\log n$ one needs to hit either $\XX_{<q,r},$ that is a clique of size less than $q$ and intersection $r$ with the planted clique, or a clique in $\Psi_{q} \cap \XX_{p,<r},$ that is a $q$-gateway of size $p$ and intersection less than $r.$ Set this target set $\mathcal{B}=(\Psi_{q} \cap \XX_{p,<r}) \cup \XX_{<q,r}$ which we hope to be ``small" enough to create a bottleneck.

We now make the following construction which has boundary included in the target set $\mathcal{B}.$ Consider $\mathcal{A}$ the set of all cliques reachable from a path with start from the empty clique and uses only cliques not included in $\mathcal{B},$ except maybe for the destination. It is easy to see $\partial A \subseteq B$ and because of the inclusion of $q$-gateways in the definition of $\mathcal{B},$ one can also see that no $q$-clique is in $\mathcal{A}$. Therefore using \cref{lem:conductance} it suffices to show that w.h.p.
\begin{align}\label{eq:goal_final_cond}
    \pi_{\beta}(\mathcal{B})/\pi_{\beta}(\mathcal{A})\leq n^{-\Omega(\log n)}.
\end{align}

To show \eqref{eq:goal_final_cond} we observe that $\XX_{p,<r} \cup \XX_{\leq p, 0} \subset \mathcal{A}$, that is $\mathcal{A}$ contains all cliques of size $p$ and intersection less than $r$ with the planted clique, or size less than $p$ and are disjoint with the planted clique. Indeed, it is straightforward than one can reach these cliques from a path from the empty clique without using cliques from $\mathcal{B}$ besides maybe the destination. 

A final calculation then gives
\begin{align}\label{eq:goal_final_cond_2}
    \pi_{\beta}(\mathcal{B})/\pi_{\beta}(\mathcal{A})\leq |\Psi_{q} \cap \XX_{p,<r}|/|\XX_{p, <r}|+\pi_{\beta}(\XX_{<q,r})/\pi_{\beta}(\XX_{ \leq p, 0})
\end{align}The first term is quasipolynomially small according to \eqref{cond_large_1}. For the second term, notice that from the equation \eqref{eq:clique_bound} from the previous subsection for all $x \leq q$ it holds $$|\XX_{x,r}|=W_{x,r} \leq n^{-\Omega(\log n)}W_{x,0}=n^{-\Omega(\log n)}|\XX_{x, 0}|.$$ Now using a first and second moment argument mostly based on \cref{basic_rg} we prove that for all $p<p' \leq q$ it holds w.h.p. $$\pi_{\beta}(\XX_{p,0}) \leq \exp({O((\beta+\log n)(|p-p'|))}\pi_{\beta}(\XX_{p', 0}).$$ Combining the above with a small case-analysis this allows us to conclude 
\begin{align*}
   \pi_{\beta}(\XX_{<q,r})/\pi_{\beta}(\XX_{ \leq p, 0}) \leq  \exp\left({O((\beta+\log n)(q-p))}\right)n^{-\Omega(\log n)}.
\end{align*}Now since $|q-p|=O(\epsilon \log n)$ and $\beta=O(\log n)$ we can always choose $\epsilon>0$ small enough but constant so that $\beta \epsilon/\log n$ is at most a small enough constant and in particular \begin{align*}
   \pi_{\beta}(\XX_{<q,r})/\pi_{\beta}(\XX_{ \leq p, 0}) \leq  n^{-\Omega(\log n)}.
\end{align*} This completes the proof overview for the failure of the Metropolis process.

For the ST dynamics notice that the only way the value of $\beta$ is used for the construction of the bottleneck is to identify a value of $\epsilon$ so that the term $\beta \epsilon/\log n$ is small enough, which then allows to choose the values of $p,q.$ But now if we have a sequence of inverse temperatures $\beta_1<\beta_2<\ldots<\beta_m$ with $\max_i |\beta_i|=O(\log n)$ we can choose a ``universal" $\epsilon$ so that for all $i$, $\left(\max_{i=1}^m \beta_i \right) \epsilon $ is small enough, leading to a ``universal" bottleneck construction for all $\pi_{\beta_i}.$ The proof then follows exactly the proof for the ST failure described in the previous subsection.

\subsection{Failure when Starting from the Empty Clique}

Here we explain the failure of the Metropolis process in the high temperature regime $\beta = o(\log n)$ when starting from the empty clique.
We show in \cref{thm:start-from-empty,thm:start-from-empty-hit-large} that, when starting from the empty clique which is the most natural and common choice of initialization, the Metropolis process fails to reach either cliques of intersection with the empty clique at least $\gamma \log n$ for any small constant $\gamma > 0$, or cliques of size $(1+\eps) \log n$ for any small constant $\eps > 0$.

One important observation is that since we are only considering when the process reaches either intersection $\gamma \log n$ or size $(1+\eps)\log n$, we may assume that the process will stop and stay at the same state once hitting such cliques. 
In particular, this means we can exclude from the state space all cliques of intersection $> \gamma \log n$ and all cliques of size $> (1+\eps)\log n$, and consider only cliques $C$ such that
\begin{equation}\label{eq:ov-size}
|C \cap \PC| \le \gamma \log n
\eqand
|C| \le (1+\eps) \log n.
\end{equation}
Indeed, the geometry of the restricted state space to these cliques are what really matters for whether the Metropolis process starting from the empty clique can reach our desired destinations or not. 
In particular, for $\gamma$ sufficiently small (say, $\gamma \le 1-\alpha$), most of cliques satisfying \cref{eq:ov-size} will have intersection $o(\log n)$ and also size $(1\pm o(1))\log n$. 
Note that this partially explains why having a large planted clique of size $n^{\alpha}$ does not help much for the tasks of interest, since for the restricted state space (those satisfying \cref{eq:ov-size}) most cliques do not have vertices from the planted clique and so any constant $\alpha<1$ does not help.

The key property that allows us to establish our result is a notion we call the ``expansion'' property.
Observe that, for a clique $C$ of size $q= \rho \log n$ with intersection $o(\log n)$, the expected number of vertices which can be added to $C$ to become a larger clique is $~n/2^q = n^{1-\rho}$; this is also the number of common neighbors of vertices from $C$.
Via the union bound and a concentration inequality, one can easily show that in fact, for all cliques with $\rho < 1$ the number of common neighbors is concentrated around its expectation with constant multiplicative error; see \cref{def:expansion,lem:satisfy-cond} (where we actually only need the lower bound).
This immediately implies that
\begin{align*}
\Pr(|X_t| = p-1 \mid |X_{t-1}| = p)
= \frac{p}{n} e^{-\beta}
\eqand
\Pr(|X_t| = p+1 \mid |X_{t-1}| = p)
\approx \frac{1}{2^p}.
\end{align*}
which allows us to track how the size of cliques evolves for the Metropolis process, under the assumption that the intersection is always $o(\log n)$. 

To actually prove our result, it will be helpful to consider the time-reversed dynamics and argue that when starting from cliques of intersection $\gamma \log n$ or cliques of size $(1+\eps) \log n$, it is unlikely to reach the empty clique. 
Suppose we have the identity Hamiltonian function $h_i = i$ for simplicity. 
Consider as an example the probability of hitting cliques of large intersection as in \cref{thm:start-from-empty}.
Recall that $\XX_{p,s}$ is the collection of cliques of size $p = \rho \log n$ for $\rho \le 1+\eps$ and of intersection $s = \gamma \log n$.
Then by reversibility we have that for all $t \ge 1$, 
\begin{align*}
\Pr\left( X_t \in \XX_{p,s} \mid X_0 = \emptyset \right)
= \sum_{\sigma \in \XX_{p,s}} \Pr\left( X_t = \sigma \mid X_0 = \emptyset \right)
= e^{\beta p} \sum_{\sigma \in \XX_{p,s}}
\Pr\left( X_t = \emptyset \mid X_0 = \sigma \right) 
\end{align*}
Notice that $e^{\beta p} = n^{o(\log n)}$ as $\beta=o(\log n)$. 
Our intuition is that in fact, for \emph{any} clique $\sigma$ of size $p$, the probability of reaching the empty clique when starting from $\sigma$ is at most $~ 1/\E[W_{p',*}]$ where $p' = \min\{p,\log n\}$; that is, for all integer $t \ge 1$ and all clique $\sigma$ of size $p$,
\begin{equation}\label{eq:ffni}
\Pr\left( X_t = \emptyset \mid X_0 = \sigma \right) 
\le \frac{1}{\E[W_{p',*}]},
\end{equation}
and hence we obtain
\[
\Pr\left( X_t \in \XX_{p,s} \mid X_0 = \emptyset \right)
\lesssim e^{\beta p} \frac{\E[W_{p,s}] }{\E[W_{p',*}]} \le n^{-c' \log n},
\]
where the last inequality is because most cliques has intersection $o(\log n)$ and size $(1+o(1)) \log n$.

The proof of \cref{eq:ffni} utilizes the expansion property mentioned above to analyze the time-reversed dynamics for $|X_t|$.
More specifically, we introduce an auxiliary birth and death process $\{Y_t\}$ on $[n]$ which is stochastically dominated by $\{|X_t|\}$, in the sense that
\begin{align*}
\Pr(Y_t = p-1 \mid Y_{t-1} = p)
&= \frac{p}{n} e^{-\beta} 
= \Pr(|X_t| = p-1 \mid |X_{t-1}| = p);\\
\Pr(Y_t = p+1 \mid Y_{t-1} = p)
&= \frac{1}{20 \cdot 2^p} < \frac{1}{2^p} \approx \Pr(|X_t| = p+1 \mid |X_{t-1}| = p).
\end{align*}
Through step-by-step coupling it is easy to see that $Y_t \le |X_t|$ always, and thus,
\begin{equation*}
\Pr\left( X_t = \emptyset \mid X_0 = \sigma \right) 
\le \Pr\left( Y_t = 0 \mid Y_0 = p \right).
\end{equation*}
This allows us to establish \cref{eq:ffni} by studying the much simpler process $\{Y_t\}$.
Furthermore, we assume the $Y_t$ process does not go beyond value $\log n$ (in fact, $(1-\eta)\log n$ for any fixed constant $\eta\in (0,1)$) so that we are in the regime where the expansion property holds; this is the reason $p' = \min\{p,\log n\}$ appears. 

The same approach works perfectly for bounding the probability of hitting cliques of size $(1+\eps)\log n$ when starting from the empty clique, as in \cref{thm:start-from-empty-hit-large}.
For the low-temperature metropolis process with $\beta = \omega(\log n)$ in \cref{thm:greedy}, we also need one more observation that the process in fact does not remove any vertex in polynomially many steps and so it is equivalent to a greedy algorithm. 
In particular, we can apply the same approach to argue that the process never reaches cliques of size $\log n$ that are \emph{subsets} of cliques with large intersection or large size, which have much smaller measure.
For the ST dynamics in \cref{thm:Simulated-Tempering-empty}, the same approach also works but we need a more sophisticated auxiliary process for the pair of clique size and inverse temperature, along with more complicated coupling arguments.

\section{Organization of Main Body}
The rest of the paper is organized as follows. In \cref{sec:prelims} we introduce the needed definitions and notation for the formal statements and proofs. In \cref{sec:large_ov} we present our lower bounds for reaching a clique of intersection $\gamma \log n$ with the planted clique. First we present the worst-case initialization result and then the case of empty clique initialization. Then in \cref{sec:bad_large} we discuss our lower bounds for reaching a clique of size $(1+\epsilon)\log n.$ As before, we first present the worst-case initialization result, and then discuss the empty clique initialization ones. Finally, in \cref{sec:st} we present our lower bounds for the simulated tempering dynamics.

\section{Getting Started}\label{sec:prelims}

For $n \in \N$, let $[n] = \{0,1,\dots,n\}$ to be the set of all non-negative integers which are at most $n$.
Throughout the paper, we use $\log$ to represent logarithm to the base $2$, i.e., $\log x = \log_2 x$ for $x \in \R^+$. 

We say an event $\mathcal{A}$ holds with high probability (w.h.p.) if $1- \Pr(\mathcal{A}) \le o(1)$. 

\subsection{Random Graphs with a Planted Clique}

Let $\GG(n,\frac{1}{2})$ denote the random graph on $n$ vertices where every pair $\{u,v\}$ is an edge with probability $1/2$ independently. 
For $k \in [n]$, we denote by $\GG(n,\frac{1}{2},k)$ the random graph $\GG(n,\frac{1}{2})$ with a planted $k$-clique, where a subset of $k$ out of $n$ vertices is chosen uniformly at random and the random graph is obtained by taking the union of $\GG(n,\frac{1}{2})$ and $\PC$ the $k$-clique formed by those vertices.

Let $\XX = \XX(G)$ be the collection of all cliques of an instance of graph $G \sim \GG(n,\frac{1}{2},k)$, and for $q,r \in [n]$ with $q \ge r$ let
\[
\XX_{q,r} = \{C \in \XX: |C| = q,\, |C \cap \PC| = r\}.
\]
We also define for convenience $\XX_{q,r} = \emptyset$ when $q < r$ or $q > n$.
Furthermore, let
\[
\XX_{q,*} = \bigcup_{r = 0}^q \XX_{q,r}
\eqand
\XX_{*,r} = \bigcup_{q = r}^n \XX_{q,r}.
\]
We also define $\XX_{\le q,*} = \bigcup_{q' = 0}^q \XX_{q',*}$ and $\XX_{*,\le r} = \bigcup_{r' = 0}^r \XX_{*,r'}$.

For $q,r \in [n]$ with $q \ge r$, let $W_{q,r} = |\XX_{q,r}|$ the number of $q$-cliques $C$ in $G$ with $|C \cap \PC|=r$. Similarly, $W_{q,r} = 0$ when $q < r$ or $q > n$.

\begin{lemma}\label{basic_rg}
Let a constant $\alpha \in [0,1)$ and consider the random graph $\GG(n,\frac{1}{2},k = \floor{n^\alpha})$ with a planted clique.
Fix any absolute constant $\epsilon>0$. 
\begin{enumerate}[(1)]
\item\label{item:X_qr-expectation} For any $q=\floor{\rho \log_2 n}$ with parameter $\rho > 0$ and any $r=\floor{\gamma \log_2 n}$ with parameter $0 \le \gamma \leq \rho$, it holds 
\begin{align*}
    \E[W_{q,r}]= \exp\left[ (\ln2) (\log n)^2 \left( \rho - \frac{\rho^2}{2} - (1-\alpha)\gamma + \frac{\gamma^2}{2} + o(1) \right) \right]
\end{align*} 
w.h.p.\ as $n \rightarrow +\infty$.

\item\label{item:X_q0>expectation} For $r=0$ and any $q=\floor{\rho \log_2 n}$ with $0<\rho \leq 2-\epsilon$, it holds 
    \begin{align*}
    W_{q,0} \geq  \frac{1}{2} \,\E[W_{q,0}] 
	\end{align*}
	w.h.p.\ as $n \rightarrow +\infty$.

\item\label{item:X_qr=0} For any $r=\floor{\gamma \log n}$ with $0 \le \gamma \le 1-\alpha$ and any $q=\floor{\rho \log_2 n}$ satisfying the inequality $\rho \ge 1+\sqrt{(1-\gamma)^2+2\alpha\gamma}+\eps$, it holds 
\begin{align*}
    W_{q,r}=0
\end{align*} 
w.h.p. as $n \rightarrow +\infty$. 
\end{enumerate} 
\end{lemma}The proof of the lemma is deferred to the Appendix.

\begin{definition}[Clique-Counts Properties $\pupp$ and $\plow$]
Let $\eps \in (0,1)$ be an arbitrary constant. 
\begin{enumerate}[(1)]
\item (Upper Bounds) We say the random graph $\GG(n,\frac{1}{2},k = \floor{n^\alpha})$ with a planted clique satisfies the property $\pupp(\eps)$ if the following is true: 
For all integers $q,r \in \N$ with $0 \le r \le q \le n$, it holds
\[
W_{q,r} \le n^3 \,\E[W_{q,r}];
\]
in particular, for $0 \le r \le (1-\alpha)\log n$ and $q \ge (1+\sqrt{(1-\gamma)^2+2\alpha\gamma}+\eps) \log n$ where $\gamma = r / \log n$, it holds
\[
W_{q,r} = 0.
\]

\item (Lower Bounds) We say the random graph $\GG(n,\frac{1}{2},k = \floor{n^\alpha})$ with a planted clique satisfies the property $\plow(\eps)$ if the following is true: 
For every integer $q \in \N$ with $0 \le q \le (2-\eps) \log n$, it holds
\[
W_{q,0} \ge \frac{1}{2} \,\E[W_{q,0}].
\]

\end{enumerate}
\end{definition}

\begin{lemma}\label{lem:rg_property1}
For any constant $\alpha \in (0,1)$ and any constant $\eps \in (0,1)$, the random graph $\GG(n,\frac{1}{2},k = \floor{n^\alpha})$ with a planted clique satisfies both $\pupp(\eps)$ and $\plow(\eps)$ with probability $1-o(1)$ as $n\to \infty$. 
\end{lemma}
\begin{proof}
Follows immediately from \cref{basic_rg}, the Markov's inequality, and the union bound. 
\end{proof}

\subsection{Hamiltonian and Gibbs Measure}

For given $n\in \N$, let $\ham: [n] \to \R$ be an arbitrary function. 
For ease of notations we write $\ham_q = \ham(q), q \in [n]$ and thus the function $\ham$ is identified by the vector $(\ham_0,\ham_1,\dots,\ham_n)$.  
Given an $n$-vertex graph $G$, 
consider the Hamiltonian function $H: \XX \to \R$ where $H(C) = \ham_{|C|}$. 
For $\beta \in \R$, the corresponding Gibbs measure is defined as
\begin{equation}\label{eq:Gibbs-general}
\pi_\beta(C) \propto w_\beta(C) := \exp\left(\beta \ham_{|C|}\right).
\end{equation}

Let $Z(\beta) = Z(G,\ham;\beta)$ to be the partition function given by
\[
Z(\beta) = \sum_{C \in \XX} w_\beta(C).
\]
Furthermore, let
\[
Z_{q,*}(\beta) = \sum_{C \in \XX_{q,*}} w_\beta(C)
\eqand
Z_{*,r}(\beta) = \sum_{C \in \XX_{*,r}} w_\beta(C).
\]


\begin{assumption}\label{ass:ham}
We assume that the Hamiltonian $h$ satisfies
\begin{enumerate}[(a)]
\item $h_0 = 0$;
\item $h$ is $1$-Lipschitz, i.e., $|h(q) - h(q')| \le |q-q'|$ for $q,q' \le 2.1 \log n$. 
\end{enumerate}
\end{assumption}

\subsection{Metropolis Process and the Hitting Time Lower Bound}\label{sec:met_process}

In this work, we study the dynamics of the Metropolis process with the respect to the Gibbs measure defined in \eqref{eq:Gibbs-general}. The Metropolis process is a Markov chain on $\XX = \XX(G)$, the space of all cliques of $G$. 
The Metropolis process is described in Algorithm \ref{alg:metropolis}.



\begin{algorithm}
\caption{Metropolis Process}\label{alg:metropolis}
\KwIn{a graph $G$, a starting clique $X_0 \in \XX$, stopping time $T$}
\For{$t = 1,\dots,T$}{
  Pick $v\in V$ uniformly at random\;
  $C \gets X_{t-1} \oplus \{v\}$\;
  \eIf{$C \in \XX$}{
    $X_t \gets 
    \begin{cases}
    C, &\text{with probability~} \min\{1, \pi_\beta(C)/\pi_\beta(X_t)\};\\
    X_{t-1}, &\text{with remaining probability};
    \end{cases}$
  }{$X_t \gets X_{t-1}$\;
  }
}
\KwOut{$X_T$}
\end{algorithm}

The Metropolis process is an ergodic and reversible Markov chain, with the unique stationary distribution $\pi_{\beta}$ except in the degenerate case when $\beta = 0$ and $G$ is the complete graph; see \cite{Jer92}.
 

The following lemma is a well-known fact for lower bounding the hitting time of some target set of a Markov chain using conductance; see \cite[Claim 2.1]{MWW09}, \cite[Theorem 7.4]{LP-book17} and \cite[Proposition 2.2]{WellsCOLT20}. We rely crucially on the following lemma for our worst-case initialization results.

\begin{lemma}\label{lem:conductance}
Let $P$ be the transition matrix of an ergodic Markov chain on a finite state space $\XX$ with stationary distribution $\pi$.
Let $A \subseteq \XX$ be a set of states and let $B \subseteq \XX$ be a set of boundary states for $A$ such that $P(x,y) = 0$ for all $x \in A \setminus B$ and $y \in A^\ccomp \setminus B$. 
Then for any $t>0$ there exists an initial state $x \in A$ such that 
\[
\Pr\left( \exists i \le t \text{~such that~} X_i \in A^\ccomp \mid X_0 = x \right) \le \frac{\pi(B)}{\pi(A)}\, t.
\]
\end{lemma}

\section{Quasi-polynomial Hitting Time of Large intersection}\label{sec:large_ov}


\subsection{Existence of a Bad Initial Clique}

\begin{theorem}\label{main_result_mixing}
Let $\alpha \in (0,1)$ be any fixed constant.
For any constant $\gamma > 0$, the random graph $\GG(n,\frac{1}{2},k = \floor{n^\alpha})$ with a planted clique satisfies the following with probability $1-o(1)$ as $n\to \infty$. 

Consider the general Gibbs measure given by \cref{eq:Gibbs-general} for arbitrary $\ham$ and arbitrary inverse temperature $\beta$.
There exists a constant $c=c(\alpha,\gamma) > 0$ and an initialization state for the Metropolis process from which it requires at least $n^{c\log n}$ steps to reach a clique of intersection with the planted clique at least $\gamma\log n$, with probability at least $1-n^{-c\log n}$. In particular, under the worst-case initialization it fails to recover the planted clique in polynomial-time.
\end{theorem} 


\begin{proof}
Notice that we can assume without loss of generality that the constant $\gamma$ satisfies $0<\gamma \le 1-\alpha$. 
We pick 
\begin{equation}\label{eq:def-eps-property}
\eps = \frac{1}{2} \left( 1 - \sqrt{(1-\gamma)^2+2\alpha\gamma} \right).
\end{equation}
Note that $\eps > 0$ since $0<\gamma \le 1-\alpha$.
Then, by \cref{lem:rg_property1} we know that the random graph $\GG(n,\frac{1}{2},\floor{n^\alpha})$ satisfies both properties $\pupp(\eps)$ and $\plow(\eps)$ simultaneously with probability $1-o(1)$ as $n\to \infty$. 
In the rest of the proof we assume that both $\pupp(\eps)$ and $\plow(\eps)$ hold.

For any $\gamma \in (0,1-\alpha]$, let $r=\floor{\gamma\log n}$. It suffices to show that there exists a constant $c=c(\alpha,\gamma)>0$ such that
\begin{align}\label{eq:goal_cond}
  \frac{ \pi_{\beta}(\XX_{*,r}) }{ \pi_{\beta}(\XX_{*,\le r}) } 
  = \frac{Z_{*,r}}{Z_{*,\le r}}
  \le \exp\left(-c \log^2 n \right).
\end{align}
Indeed, given \cref{eq:goal_cond}, \cref{main_result_mixing} is an immediate consequence of \cref{lem:conductance}.

By the property $\pupp(\eps)$ we have that $W_{q,r} = 0$ for all $q > \bar{q}$ where
\[
\bar{q} = \floor{1+\sqrt{(1-\gamma)^2+2\alpha\gamma}+\eps} \log n
= \floor{2-\eps} \log n.
\]
Hence, we get from the definition of the restricted partition function $Z_{*,r}$ that
\[
Z_{*,r}
= \sum_{q = r}^n Z_{q,r}
= \sum_{q = r}^n W_{q,r} \exp\left( \beta\ham_q \right)
= \sum_{q = r}^{\bar{q}} W_{q,r} \exp\left( \beta\ham_q \right)
\le n^3 \sum_{q = r}^{\bar{q}} \E\left[ W_{q,r} \right] \exp\left( \beta\ham_q \right),
\]
where the last inequality again follows from $\pupp(\eps)$.
We define
\[
q^* = \argmax_{q \in [\bar{q}]}\, \E\left[ W_{q,r} \right] \exp\left( \beta\ham_q \right).
\] 
Thus, we have that
\begin{equation}\label{eq:Z*r}
Z_{*, r} \le n^4 \,\E\left[ W_{q^*,r} \right] \exp\left( \beta\ham_{q^*} \right).
\end{equation}

Meanwhile, we have 
\[
Z_{*,\le r} \ge Z_{q^*,0} = W_{q^*,0} \exp\left( \beta\ham_{q^*} \right).
\]
We deduce from $q^* \le \bar{q} \le (2-\eps)\log n$ and the property $\plow(\eps)$ that 
\begin{equation}\label{eq:Z*ler}
Z_{*,\le r} \ge \frac{1}{2} \,\E\left[ W_{q^*,0} \right] \exp\left( \beta\ham_{q^*} \right).
\end{equation}

Combining \cref{eq:Z*r,eq:Z*ler}, we get that,
\[
\frac{Z_{*,r}}{Z_{*,\le r}} \le 2n^4 \,\frac{\E\left[ W_{q^*,r} \right]}{\E\left[ W_{q^*,0} \right]}.
\]
Suppose that $\rho = q^*/\log n$.
Then, an application of \cref{item:X_qr-expectation} of \cref{basic_rg} implies that,
\begin{align*}
\frac{Z_{*,r}}{Z_{*,\le r}} &\le 
2n^4 \,\frac{\exp\left[ (\ln2) (\log n)^2 \left( \rho - \frac{\rho^2}{2} - (1-\alpha)\gamma + \frac{\gamma^2}{2} + o(1) \right) \right]}{\exp\left[ (\ln2) (\log n)^2 \left( \rho - \frac{\rho^2}{2} + o(1) \right) \right]} \\
&\le \exp\left[ (\ln2) (\log n)^2 \left( - (1-\alpha)\gamma + \frac{\gamma^2}{2} + o(1) \right) \right].
\end{align*}This establishes \eqref{eq:goal_cond} for $c=c(\alpha,\gamma):= (1-\alpha)\gamma - \frac{\gamma^2}{2}>0$ since $0 < \gamma \le 1-\alpha$, as we wanted.
\end{proof}

\subsection{Starting from the Empty Clique}

In this section, we strengthen \cref{main_result_mixing} for a wide range of temperatures by showing that the Metropolis Dynamics still fails to obtain a significant intersection with the planted clique when starting from the empty clique (or any clique of sufficiently small size), a nature choice of initial configuration in practice.

\subsubsection{Our Result}

\begin{theorem}\label{thm:start-from-empty}
Let $\alpha \in (0,1)$ be any fixed constant.
For any constant $\gamma \in (0,1-\alpha)$, the random graph $\GG(n,\frac{1}{2},k = \floor{n^\alpha})$ with a planted clique satisfies the following with probability $1-o(1)$ as $n\to \infty$. 

Consider the general Gibbs measure given by \cref{eq:Gibbs-general} for arbitrary $1$-Lipschitz $\ham$ with $h_0 = 0$ and inverse temperature 
$\beta \le (\ln2) \gamma \log n$. 
Let $\{X_t\}$ denote the Metropolis process on $G$ with stationary distribution $\pi_\beta$. 
Then there exist constants $\rhozero = \rhozero(\alpha,\gamma)>0$ and $c = c(\alpha,\gamma)>0$ such that for any clique $C \in \XX$ of size at most $\rhozero\log n$, one has 
\[
\Pr\Big(\exists t \in \N \wedge t \le n^{c \log n} \text{~s.t.~} |X_t \cap \PC| \ge \gamma \log n \;\Big\vert\; X_0 = C \Big) \le  n^{-c \log n}.
\]
In particular, the Metropolis process starting from the empty clique requires $n^{\Omega(\log n)}$ steps to reach cliques of intersection with the planted clique at least $\gamma\log n$, with probability $1 - n^{-\Omega(\log n)}$. 
As a consequence, it fails to recover the planted clique in polynomial time.
\end{theorem}
 We proceed with the proof of the theorem.

\subsubsection{Key Lemmas}
We start with tracking how the size of the clique changes during the process.
It is also helpful to consider the reversed process: show that it is unlikely to hit the empty clique $\emptyset$ when starting from some clique of size $\log n$.

\medskip


\begin{definition}
For a graph $G=(V,E)$ and a subset $U \subseteq V$ of vertices, we say a vertex $v \in V \setminus U$ is \emph{fully adjacent} to $U$ if $v$ is adjacent to all vertices in $U$; equivalently, $U \subseteq N(v)$ where $N(v)$ denotes the set of neighboring vertices of $v$ in the graph $G$. 
Let $A(U)$ denote the set of all vertices in $V \setminus U$ that are fully adjacent to $U$; equivalently, $A(U)$ is the set of all common neighbors in $V \setminus U$ of all vertices from $U$.
\end{definition}


\begin{definition}[Expansion Property $\pexp$]
\label{def:expansion}
Let $\eta \in (0,1)$ be an arbitrary constant.
We say the random graph $\GG(n,\frac{1}{2},k = \floor{n^\alpha})$ with a planted clique satisfies the expansion property $\pexp(\eta)$ if the following is true:
For every $U \subseteq V$ with $|U| \le (1-\eta) \log n$, it holds 
\[
|A(U)| \ge \frac{n}{20 \cdot 2^{|U|}}.
\] 
\end{definition}

The following lemma establishes the desired ``expansion lemma" on the cliques of size less than $\log n$ in $G(n,\frac{1}{2},k).$

\begin{lemma}[``Expansion Lemma"]
\label{lem:satisfy-cond}
For any constant $\alpha \in (0,1)$ and any constant $\eta \in (0,1)$, the random graph $\GG(n,\frac{1}{2},k = \floor{n^\alpha})$ with a planted clique satisfies the expansion property $\pexp(\eta)$ with probability $1-o(1)$ as $n\to \infty$. 
\end{lemma}The proof of the lemma is deferred to the Appendix.


The \cref{lem:satisfy-cond} is quite useful for us, as it allows us to obtain the following bounds on the size transitions for the Metropolis process:
\begin{align}
\Pr\left(|X_t| = q-1 \mid |X_{t-1}| = q \right) &= \frac{q}{n} \min \left\{ 1, \exp\left[ \beta \left( \ham_{q-1} - \ham_q \right) \right] \right\} \label{eq:size-bound}\\
\Pr\left(|X_t| = q+1 \mid |X_{t-1}| = q \right) &\ge \frac{1}{20 \cdot 2^q} \min \left\{ 1, \exp\left[ \beta \left( \ham_{q+1} - \ham_q \right) \right] \right\}. \label{eq:size-bound2}
\end{align}


The proof is also making use of the following delicate lemma.

\begin{lemma}\label{lem:reaching-prob}
Consider the random graph $\GG(n,\frac{1}{2},k = \floor{n^\alpha})$ with a planted clique conditional on satisfying the property $\pupp(0.1)$ and the expansion property $\pexp(\eta)$ for some fixed constant $\eta \in (0,1)$. 
Let $t,p,q,r \in \N_+$ be integers with $p\le r \leq q$. Denote also $\rhozero= \qzero/\log n,$ $\rho=q/\log n$ and $\gamma=r/\log n.$
For any $C \in \Omega_{\qzero,*}$ we have
\begin{multline*}
 \Pr\left( X_t \in \XX_{q,r} \mid X_0 = C \right) 
\le \\
 t\exp\left[ (\ln 2) (\log n)^2 \left( \left( (1+\hat{\beta})\rho - \frac{\rho^2}{2} \right) - \left( (1+\hat{\beta})\rho' - \frac{(\rho')^2}{2} \right) - \left( (1-\alpha)\gamma - \frac{\gamma^2}{2} - \rhozero + \frac{\rhozero^2}{2} \right) + o(1) \right) \right], 
\end{multline*}
where $\rho' = \min\{\rho, 1-\eta\}$ and $\hat{\beta}=\beta/((\ln 2) (\log n)).$
\end{lemma} 
We postpone the proof of \cref{lem:reaching-prob} to \cref{subsec:reach-prob-proof} and first show how it can be used to prove \cref{thm:start-from-empty}.
On a high level, we bound the probability of hitting $\XX_{q,r}$ by studying how the size of the clique evolves during the process up to size $(1-\eta)\log n$. 
The term $( (1+\hat{\beta})\rho - \rho^2/2 ) - ( (1+\hat{\beta})\rho' - (\rho')^2/2 )$ represents the approximation error when the clique goes beyond size $(1-\eta)\log n$; in particular, when the destination clique size $q = \rho \log n$ is at most $(1-\eta) \log n$, we have $\rho' = \rho$ and this error is zero. 
Meanwhile, the term $(1-\alpha)\gamma - \gamma^2/2 - \rhozero + \rhozero^2/2$ corresponds to the fact that the number of cliques of size $q$ and intersection $r = \gamma \log n$ with the planted clique is much smaller than the total number of cliques of size $q$, and hence reaching intersection $r$ is very unlikely.

\subsubsection{Proof of \cref{thm:start-from-empty}, given \cref{lem:reaching-prob}}
\begin{proof}[Proof of \cref{thm:start-from-empty}]

Let $\hat{\beta}=\beta/ ((\ln2) (\log n))$ and recall $\hat{\beta} \le \gamma < 1-\alpha.$
As will be clear later, we shall choose
\[
\xi = \frac{1}{4}(1-\alpha-\gamma)\gamma
\eqand
\eta = \min\left\{ \frac{1}{8}(1-\alpha-\gamma), \gamma \right\}.
\]

By \cref{lem:rg_property1,lem:satisfy-cond}, the random graph $\GG(n,\frac{1}{2},\floor{n^\alpha})$ satisfies both $\pupp(0.1)$ and $\pexp(\eta)$ with probability $1-o(1)$ as $n\to \infty$, for the choice of $\eta$ given above. 
In the rest of the proof we assume that both $\pupp(0.1)$ and $\pexp(\eta)$ are satisfied.

Suppose that $|C| = \qzero =\rhozero' \log n \leq \rhozero \log n$. 
By the union bound we have
\begin{align}
\Pr\Big( \exists t \in \N \wedge t \le T:\, X_t \in \XX_{*,r} \;\Big\vert\; X_0 = C \Big) \nonumber
&\le 
\sum_{t = 1}^T \sum_{q = r}^n \Pr\left( X_t \in \XX_{q,r} \mid X_0 = C \right) \nonumber\\
&\le
Tn \,\max_{t \in [T]} \,\max_{q \in [n]} \,\Pr\left( X_t \in \XX_{q,r} \mid X_0 = C \right) \label{eq:final_ub}.
\end{align}

Now let us fix some $t \in [T], q \in [n].$ By \cref{lem:reaching-prob},
\begin{align*}
& \Pr\left( X_t \in \XX_{q,r} \mid X_0 = C \right) \\
\le{}& t \exp\left[ (\ln2) (\log n)^2 \left( \left( (1+\hat{\beta})\rho - \frac{\rho^2}{2} \right) - \left( (1+\hat{\beta})\rho' - \frac{(\rho')^2}{2} \right) - \left( (1-\alpha)\gamma - \frac{\gamma^2}{2} - \rhozero' + \frac{(\rhozero')^2}{2} \right) + o(1) \right) \right]\\
\le{}& t \exp\left[ (\ln2) (\log n)^2 \left( \left( (1+\hat{\beta})\rho - \frac{\rho^2}{2} \right) - \left( (1+\hat{\beta})\rho' - \frac{(\rho')^2}{2} \right) - \left( (1-\alpha)\gamma - \frac{\gamma^2}{2} - \rhozero + \frac{\rhozero^2}{2} \right) + o(1) \right) \right],
\end{align*}where we have used that $\rhozero'\leq \rhozero \leq 1.$
Write for shorthand
\[
A = \left( (1+\hat{\beta})\rho - \frac{\rho^2}{2} \right) - \left( (1+\hat{\beta})\rho' - \frac{(\rho')^2}{2} \right)
\quad\text{and}\quad
B = (1-\alpha)\gamma - \frac{\gamma^2}{2} - \rhozero + \frac{\rhozero^2}{2}.
\]
So we have
\begin{align}\label{eq:almost_final}
\Pr\left( X_t \in \XX_{q,r} \mid X_0 = C \right) \le t \exp\left[ (\ln2) (\log n)^2 \left( A - B + o(1) \right) \right].
\end{align}
Given \eqref{eq:almost_final} it suffices to show that for all $\alpha \in (0,1), \gamma \in (0,1-\alpha)$ there exists $c_0(\alpha,\gamma)>0$ such that uniformly for all values of interest of the parameters $\rho,\hat{\beta}$ we have   $A - B \leq -c_0(\alpha,\gamma)$. Indeed, then by \eqref{eq:final_ub} we have
\begin{align*}
\Pr\Big( \exists t \in \N \wedge t \le T:\, X_t \in \XX_{*,r} \;\Big\vert\; X_0 = C \Big) \leq T^2 n \exp\left[ -c_0(\alpha,\gamma)(\ln2) (\log n)^2  \right]
\end{align*} and \cref{thm:start-from-empty} follows e.g. for $c(\alpha,\gamma)=\frac{\ln 2 }{20}c_0(\alpha,\gamma).$

We now construct the desired function $c_0(\alpha,\gamma)>0.$
If $\rho \le 1-\eta$, then $\rho' = \rho$ and $A = 0$. 
Meanwhile, we have
\[
B = (1-\alpha)\gamma - \frac{\gamma^2}{2} - \rhozero + \frac{\rhozero^2}{2}
\ge (1-\alpha)\gamma - \frac{\gamma^2}{2} - \frac{1}{4}(1-\alpha-\gamma)\gamma 
\ge \frac{3}{4}(1-\alpha-\gamma)\gamma.
\]
So $A-B \le -\frac{3}{4}(1-\alpha-\gamma)\gamma$. 

If $\rho > 1-\eta$, then $\rho' = 1-\eta$ and we have
\begin{align*}
A &= \left( (1+\hat{\beta})\rho - \frac{\rho^2}{2} \right) - \left( (1+\hat{\beta})\rho' - \frac{(\rho')^2}{2} \right) \\
&= \hat{\beta}(\rho-1) - \frac{1}{2}(\rho-1)^2 + \hat{\beta}\eta + \frac{\eta^2}{2} \\
&\le \frac{\gamma^2}{2} + \frac{1}{4}(1-\alpha-\gamma)\gamma,
\end{align*}
where the last inequality follows from 
\[
\hat{\beta}(\rho-1) - \frac{1}{2}(\rho-1)^2
\le \frac{\hat{\beta}^2}{2} \le \frac{\gamma^2}{2}
\]
and also
\[
\hat{\beta}\eta + \frac{\eta^2}{2}
\le \gamma \cdot \frac{1}{8}(1-\alpha-\gamma) + \frac{\gamma}{2} \cdot \frac{1}{8}(1-\alpha-\gamma)
\le \frac{1}{4}(1-\alpha-\gamma)\gamma
\]
since $\eta \le (1-\alpha-\gamma)/8$ and $\eta \le \gamma$.
Meanwhile, 
we have
\[
B = (1-\alpha)\gamma - \frac{\gamma^2}{2} - \rhozero + \frac{\rhozero^2}{2}
\ge (1-\alpha)\gamma - \frac{\gamma^2}{2} - \frac{1}{4}(1-\alpha-\gamma)\gamma. 
\]
So, we deduce that
\[
A - B
\le \frac{\gamma^2}{2} + \frac{1}{4}(1-\alpha-\gamma)\gamma
- (1-\alpha)\gamma + \frac{\gamma^2}{2} + \frac{1}{4}(1-\alpha-\gamma)\gamma
= - \frac{1}{2}(1-\alpha-\gamma)\gamma.
\]


Hence, in all cases $A-B \leq -c_0(\alpha, \gamma)$ for 
$$
c_0(\alpha,\gamma) = \frac{1}{2}(1-\alpha-\gamma)\gamma.
$$ 
This completes the proof of the theorem.
\end{proof}

\subsubsection{Proof of \cref{lem:reaching-prob}} 
\label{subsec:reach-prob-proof}

In this subsection we prove the crucial \cref{lem:reaching-prob}.  
Throughout this subsection we assume that the property $\pupp(0.1)$ and the expansion property $\pexp(\eta)$ hold for some fixed constant $\eta \in (0,1)$. 
First, recall that $W_{q,r} = |\XX_{q,r}|$. 
We have from $\pupp(0.1)$ that
\begin{align}
\Pr\left( X_t \in \XX_{q,r} \mid X_0 = C \right)
&= \sum_{\sigma \in \XX_{q,r}} \Pr\left( X_t = \sigma \mid X_0 = C \right) \nonumber\\
&\le W_{q,r} \max_{\sigma \in \XX_{q,r}} \Pr\left( X_t = \sigma \mid X_0 = C \right) \nonumber\\
&\le n^3 \,\E[W_{q,r}] \max_{\sigma \in \XX_{q,r}} \Pr\left( X_t = \sigma \mid X_0 = C \right) \label{first_ub}.
\end{align} For now, we fix a $\sigma \in \XX_{q,r}$ and focus on bounding $\Pr\left( X_t = \sigma \mid X_0 = C \right).$ The key idea to bound this probability is to exploit the reversibility of the Metropolis process. The following two standard facts are going to be useful.

\begin{fact}[\cite{LP-book17}]
\label{fact:reversibility} 
If $P$ is the transition matrix of a \emph{reversible} Markov chain over a finite state space $\Gamma$ with stationary distribution $\mu$, then for all $x,y \in \Gamma$ and all integer $t \ge 1$ it holds
\[
\mu(x) P^t(x,y) = \mu(y) P^t(y,x).
\]
\end{fact}

\begin{fact}[\cite{LP-book17}]\label{fact:bd-stationary}
For a birth-death process on $[n]$ with transition probabilities
\[
P(i, i+1) = p_i,
\quad
P(i, i-1) = q_i,
\quad\text{and}\quad
P(i,i) = 1 - p_i - q_i,
\]
the stationary distribution is given by
\[
\mu(i) \propto \prod_{s = 1}^{i} \frac{p_{s-1}}{q_s}, \quad \forall i \in [n].
\]
\end{fact}

Now, notice that using the time-reversed dynamics it suffices try to bound the probability of reaching a small clique $C$ when starting from a large clique $\sigma$. Indeed, by reversibility we have
\begin{equation}\label{eq:reversible-eq}
\Pr\left( X_t = \sigma \mid X_0 = C \right)
= \exp\left( \beta \left( \ham_q - \ham_{\qzero} \right) \right) \Pr\left( X_t = C \mid X_0 = \sigma \right).
\end{equation}
which is an application of \cref{fact:reversibility}. 

We introduce a birth-death process on $[n]$ denoted by $\{Y_t\}$ with transition matrix $P$ given by the following transition probabilities:
\begin{align*}
P(\size,\size-1) &= \frac{\size}{n} \min \left\{\exp\left( \beta \left( \ham_{\size-1} - \ham_{\size} \right) \right), 1 \right\}, \quad 1 \le \size \le n; \\
P(\size,\size+1) &=
\begin{cases}
\frac{1}{20 \cdot 2^\size} \min \left\{\exp\left( \beta \left( \ham_{\size+1} - \ham_{\size} \right) \right), 1 \right\}, &\quad 0\le \size < \floor{(1-\eta) \log n}; \\
0, &\quad \floor{(1-\eta) \log n} \le \size \le n-1;
\end{cases}\\
P(\size,\size) &= 1 - P(\size,\size-1) - P(\size,\size+1), \quad 0 \le \size \le n.
\end{align*}
Denote the stationary distribution of $\{Y_t\}$ by $\nu$, which is supported on $\{0,1,\dots, \floor{(1-\eta) \log n}\}$. 
The process $\{Y_t\}$ serves as an approximation of $\{|X_t|\}$; note that $\{|X_t|\}$ itself is not a Markov process. 

The following lemma shows that $Y_t$ is stochastically dominated by $|X_t|$. The proof of this fact is essentially based on the expansion property $\pexp(\eta)$, and the derived in the proof below bounds on the size transition of $|X_t|$, described in \cref{eq:size-bound,eq:size-bound2}.

\begin{lemma}\label{lem:stochastic-dominance}
Let $\{X_t\}$ denote the Metropolis process starting from some $X_0 = \sigma \in \XX_{q,*}$.
Let $\{Y_t\}$ denote the birth-death process described above with parameter $\eta \in (0,1)$ starting from $Y_0 = q$.
Then there exists a coupling $\{(X_t,Y_t)\}$ of the two processes such that for all integer $t \ge 1$ it holds
\[
Y_t \le |X_t|.
\]
In particular, for all integer $t \ge 1$ it holds
\[
\Pr\left( X_t = C \mid X_0 = \sigma \right) \le \Pr\left( \exists t' \in \N \wedge t' \le t:\, Y_{t'} = \qzero \mid Y_0 = q \right)
\le \sum_{t'=1}^t \Pr\left( Y_{t'} = \qzero \mid Y_0 = q \right).  
\]


\end{lemma}
\begin{proof}
We couple $\{|X_t|\}$ and $\{Y_t\}$ as follows. 
Suppose that $Y_{t-1}\le |X_{t-1}|$ for some integer $t \ge 1$.
We will construct a coupling of $X_t$ and $Y_t$ such that $Y_t \le |X_t|$. Notice that the following probability inequality is a straightforward corollary of that.
 
 Since the probability that $Y_t = Y_{t-1} + 1$ is less than $1/2$ and so does the probability of $|X_t| = |X_{t-1}| - 1$, we may couple $X_t$ and $Y_t$ such that $|X_t| - Y_t$ decreases at most one; namely, it never happens that $Y_t$ increases by $1$ while $X_t$ decreases in size. 
Thus, it suffices to consider the extremal case when $|X_{t-1}| = Y_{t-1} = \size$. We have
\begin{align}
\Pr\left( |X_t| = \size - 1 \mid |X_{t-1}| = \size \right)
&= \frac{\size}{n} \min\left\{ 1, \exp\left[ \beta \left( \ham_{\size-1} - \ham_{\size} \right) \right] \right\} = P\left( \size, \size-1 \right) \label{couple_down}
\end{align}

Meanwhile, recall that $A(X_{t-1})$ is the set of vertices $v$ such that $X_{t-1} \cup \{v\} \in \XX$. Then we have 
\[
|A(X_{t-1})| \ge \frac{n}{20\cdot 2^\size}
\]
whenever $\size \le n_\eta := \floor{(1-\eta)\log n}$ by the expansion property $\pexp(\eta)$. Hence, we deduce that
\begin{align}
\Pr\left( |X_t| = \size + 1 \mid |X_{t-1}| = \size \right)
&= \dfrac{|A(X_{t-1})|}{n} \min \left\{ 1, \exp\left[ \beta \left( \ham_{\size+1} - \ham_\size \right) \right] \right\} \nonumber \\
&\ge \dfrac{\one\{\size < n_\eta\}}{20 \cdot 2^\size} \min \left\{ 1, \exp\left[ \beta \left( \ham_{\size+1} - \ham_\size \right) \right] \right\}
= P\left( \size, \size+1 \right) \label{couple_up}
\end{align}
Using \eqref{couple_down} and \eqref{couple_up} we can couple $|X_t|$ and $Y_t$ such that either $|X_t| = Y_t$ or $|X_t| = \size+1$ and $Y_t = \size$, as desired.
\end{proof}

The next lemma upper bounds the $t$-step transition probability $\Pr\left(Y_t = \qzero \mid Y_0 = q\right)$.

\begin{lemma}\label{lem:Y-ub-prob}
Let $\{Y_t\}$ denote the birth-death process described above with parameter $\eta \in (0,1)$ starting from $Y_0 = q = \rho \log n$ and let $p=\rhozero \log n$ with $\rhozero \leq 1-\eta$.
Then for all integer $t \ge 1$ we have
\begin{align}\label{eq:up_prob_Y}
\Pr\left(Y_t = \qzero \mid Y_0 = q \right) 
\le \exp\left[ (\ln2) (\log n)^2 \left( - \rho' + \frac{(\rho')^2}{2} + \rhozero - \frac{\rhozero^2}{2} + o(1) \right) \right] \exp\left[ \beta \left( \ham_{\qzero} - \ham_{q'} \right) \right]
\end{align}
where $\rho' = \min\{ \rho, 1-\eta \}$ and $q' = \rho' \log n$. 
\end{lemma}
\begin{proof}
We consider first the case where $\rho \leq 1-\eta$.  By \cref{fact:reversibility}, we have
\begin{align}
\Pr\left(Y_t = \qzero \mid Y_0 = q \right) 
= P^t(q, \qzero) 
= \frac{\nu(\qzero)}{\nu(q)} P^t(\qzero,q) 
\le \frac{\nu(\qzero)}{\nu(q)}. \label{eq:station_bound}
\end{align}
By \cref{fact:bd-stationary}, we have
\begin{align*}
\frac{\nu(\qzero)}{\nu(q)} 
&= \prod_{\size = \qzero+1}^q \frac{ \frac{\size}{n} \min \left\{ \exp\left[ \beta \left( \ham_{\size-1} - \ham_\size \right) \right], 1 \right\} }{ \frac{1}{20 \cdot 2^{\size-1}} \min \left\{ \exp\left[ \beta \left( \ham_\size - \ham_{\size-1} \right) \right], 1 \right\} } \\
&= \prod_{\size = \qzero+1}^q \frac{20\size \cdot 2^{\size-1}}{n} \exp\left[ \beta \left( \ham_{\size-1} - \ham_\size \right) \right] \\
&= \frac{20^{q-\qzero} \cdot (q! / \qzero !) \cdot 2^{\binom{q}{2} - \binom{\qzero}{2}}}{n^{q-\qzero}} \exp\left[ \beta \left( \ham_{\qzero} - \ham_q \right) \right] \\
&= \exp\left[ (\ln2) (\log n)^2 \left( - \rho + \frac{\rho^2}{2} + \rhozero - \frac{\rhozero^2}{2} + o(1) \right) \right] \exp\left[ \beta \left( \ham_{\qzero} - \ham_q \right) \right].
\end{align*}

Next, consider the case where $\rho > 1-\eta$, or equivalently $q>q'$. 
Let $\tau$ be the first time that $Y_{t'} = q'$ and we obtain from \eqref{eq:station_bound} that
\begin{align*}
    \Pr\left(Y_t = \qzero \mid Y_0 = q \right)  &=\sum_{t'=0}^t  \Pr\left(\tau = t' \mid Y_0 = q \right) \Pr\left(Y_{t-t'} = \qzero \mid Y_0 = q' \right) \\
    &\le \frac{\nu(\qzero)}{\nu(q')} \sum_{t'=0}^t  \Pr\left(\tau = t' \mid Y_0 = q \right) 
    = \frac{\nu(\qzero)}{\nu(q')} \Pr\left(\tau \le t \mid Y_0 = q \right)
    \le \frac{\nu(\qzero)}{\nu(q')}.
\end{align*}This completes the proof of the lemma.
\end{proof}

We now proceed with the proof of \cref{lem:reaching-prob}.
\begin{proof}[Proof of \cref{lem:reaching-prob}]

From \cref{lem:stochastic-dominance,lem:Y-ub-prob}, we have for each $\sigma \in \XX_{q,r}$ \begin{align*}
\Pr\left( X_t = \sigma \mid X_0 = C \right) 
&= \exp\left[ \beta \left( \ham_q - \ham_{\qzero} \right) \right] \Pr\left( X_t = C \mid X_0 = \sigma \right) \\
&\le \exp\left[ \beta \left( \ham_q - \ham_{\qzero} \right) \right] \,\max_{t' \in [t]}\, t\Pr\left(Y_{t'} = \qzero \mid Y_0 = q\right) \\
&\le t  \exp\left[ (\ln2) (\log n)^2 \left( - \rho' + \frac{(\rho')^2}{2} + \rhozero - \frac{\rhozero^2}{2} + o(1) \right) \right] \exp\left[ \beta \left( \ham_q - \ham_{q'} \right) \right],
\end{align*}
where $\rho' = \min\{ \rho, 1-\eta \}$ and $q' = \rho' \log n$.

Combining now with \cref{first_ub} and \cref{basic_rg} we have that $\Pr\left( X_t \in \XX_{q,r} \mid X_0 = C \right)$ is at most
\begin{align*}
& n^3\, \E[W_{q,r}] \max_{\sigma \in \XX_{q,r}} \Pr\left( X_t = \sigma \mid X_0 = C \right) \\
\le{}& n^3 t \exp\left[ (\ln2) (\log n)^2 \left( \rho - \frac{\rho^2}{2} - (1-\alpha)\gamma + \frac{\gamma^2}{2} +o(1) \right) \right] \\
&\cdot \exp\left[ (\ln2) (\log n)^2 \left( - \rho' + \frac{(\rho')^2}{2} + \rhozero - \frac{\rhozero^2}{2} + o(1) \right) \right] \exp\left[ \beta \left( \ham_q - \ham_{q'} \right) \right]  \\
={}& t \exp\left[ (\ln2) (\log n)^2 \left( \left( \rho - \frac{\rho^2}{2} \right) - \left( \rho' - \frac{(\rho')^2}{2} \right) - \left( (1-\alpha)\gamma - \frac{\gamma^2}{2} - \rhozero + \frac{\rhozero^2}{2} \right) +o(1) \right) \right] \exp\left[ \beta \left( \ham_q - \ham_{q'} \right) \right].
\end{align*}
Since $h$ is $1$-Lipschitz and $q' \le q$, we have
\[
\exp\left[ \beta \left( \ham_q - \ham_{q'} \right) \right] \le \exp\left[ \beta \left( q - q' \right) \right]
\le \exp\left[ (\ln2) (\log n)^2 \left( \hat{\beta} \left( \rho - \rho' \right) \right) \right]
\]
where we recall that $ \beta = (\ln 2) \hat{\beta} \log n$.
Therefore,
\begin{align*}
& \Pr\left( X_t \in \XX_{q,r} \mid X_0 = C \right) \\
\le{}& t \exp\left[ (\ln2) (\log n)^2 \left( \left( (1+\hat{\beta})\rho - \frac{\rho^2}{2} \right) - \left( (1+\hat{\beta})\rho' - \frac{(\rho')^2}{2} \right) - \left( (1-\alpha)\gamma - \frac{\gamma^2}{2} - \rhozero + \frac{\rhozero^2}{2} \right) + o(1) \right) \right].
\end{align*} 
The proof of the lemma is complete.
\end{proof}

\section{Quasi-polynomial Hitting Time of Large Cliques}\label{sec:bad_large}

In this section, we present our results about the failure of the Metropolis process to even find cliques of size at least $(1+\epsilon)\log n$, for any planted clique size $k=\floor{n^{\alpha}}, \alpha \in (0,1).$ 

\subsection{Existence of a Bad Initial Clique}

We start with the ``worst-case" initialization result which now works for all inverse temperatures $\beta=O(\log n)$, but establishes that from this initialization the Metropolis process fails to find either a clique of size at least $(1+\epsilon)\log n$ or to find a clique with intersection at least $\gamma \log n$ with the planted clique.

\begin{theorem}\label{thm:large-clique}
Let $\alpha \in [0,1)$ be any fixed constant.
Then the random graph $\GG(n,\frac{1}{2},k = \floor{n^\alpha})$ with a planted clique satisfies the following with probability $1-o(1)$ as $n\to \infty$. 

Consider the general Gibbs measure given by \cref{eq:Gibbs-general} for arbitrary $\ham$ satisfying \cref{ass:ham} and arbitrary inverse temperature $\beta = O(\log n)$. 
For any constants $\eps \in (0,1-\alpha)$ and $\gamma \in (0,1-\alpha]$, there exists a constant $c > 0$ and an initialization state for the Metropolis process from which it requires at least $n^{c\log n}$ steps to reach 
\begin{itemize}
    \item either cliques of size at least $(1+\eps)\log n$, 
    \item or cliques of intersection with the planted clique at least $\gamma \log n$,
\end{itemize}
with probability at least $1-n^{-c\log n}$.
\end{theorem}

We now present the proof of \cref{thm:large-clique}. We first need the notion of \emph{gateways} as introduced by \cite{Jer92} in his original argument for the failure of the Metropolis process.

\begin{definition}[Gateways]\label{def:gateway}
For $q\in [n]$, we say a clique $C \in \XX$ is a $q$-gateway if there exists $\ell \in \N$ and a sequence of cliques $C_0 = C, C_1, \dots, C_\ell \in \XX$ such that
\begin{enumerate}[(1)]
\item For every $1 \le i \le \ell$, $C_{i-1}$ and $C_i$ differ by exactly one vertex;
\item For every $0 \le i \le \ell$, $|C_{i}| \ge |C|$;
\item $|C_\ell| = q$.
\end{enumerate}
Let $\Psi_q$ denote the collection of all cliques that are $q$-gateways.
\end{definition}
Notice that by definition if a clique $\sigma$ is a $q$-gateway then $|\sigma| \le q$.

\begin{definition}[Gateway-Counts Property $\pgw$]
We say the random graph $\GG(n,\frac{1}{2},k = \floor{n^\alpha})$ with a planted clique satisfies the gateway-counts property $\pgw$ if the following is true:
For all integers $q = \floor{(1+\eps) \log n}$, $p = \floor{(1+\eps-\theta) \log n}$, and $u = \floor{(\eps/6) \log n}$ with parameters $\eps \in (0,1-\alpha)$ and $\theta \in (0,\eps)$,
it holds 
\[
\left| \Psi_q \cap \Omega_{p, \le u} \right|
\le \exp\left[ (\ln2) (\log n)^2 \left( (1+\eps-\theta) - \frac{1}{2}(1+\eps-\theta)^2 -\theta \left( \frac{5}{6}\eps-2\theta \right) + o(1) \right) \right].
\]
\end{definition}

The following lemma follows immediately from arguments in \cite{Jer92}.

\begin{lemma}\label{lem:gateway}
For any constant $\alpha \in [0,1)$, the random graph $\GG(n,\frac{1}{2},k = \floor{n^\alpha})$ with a planted clique satisfies the gateway-counts property $\pgw$ with probability $1-o(1)$ as $n\to \infty$. 
\end{lemma}

\begin{proof}
We follow the same approach as in \cite{Jer92} with the slight modification that $\theta$ is not equal to $\eps/3$ but arbitrary. 
For any $C \in \Psi_q \cap \Omega_{p, \le u}$, there exists a set $U \in V \setminus C$ of size $|U| = q-p$ and a subset $W \subseteq C \setminus \PC$ of size $2p-q-u$ such that every vertex from $U$ is adjacent to every vertex from $W$. To see this, consider a path $C_0 = C, C_1, \dots, C_\ell$ as in \cref{def:gateway}, and consider the first clique $C'$ in the path such that $|C' \setminus C| = q-p$. Such $C'$ must exist since the destination clique $C_\ell$ has size $q$ while $|C| = p < q$. Note that $C'$ corresponds to the first time when $q-p$ new vertices are added. Meanwhile, since $|C'| \ge p$, we have $|C \cap C'| \ge p-(q-p) = 2p-q$ and $|C \cap C' \setminus \PC| \ge 2p-q-u$. We can thus take $U = C' \setminus C$ and any $W \subseteq C \cap C' \setminus \PC$ of size $2p-q-u$.

Hence, we can associate every $q$-gateway $C$ in $\Omega_{p, \le u}$ with a tuple $(C,U,W)$ satisfying all the conditions mentioned above: $C$ is a clique of size $p$ and intersection at most $u$ with $\PC$, $U \subseteq V \setminus C$ has size $q-p$, $W\subseteq C \setminus \PC$ has size $2p-q-u$, and $U,W$ are fully connected.
Let $X$ denote the number of such tuples. Then, $|\Psi_q \cap \Omega_{p, \le u}| \le X$.
The first moment of $X$ is given by
\begin{align*}
\E[X] &= \sum_{r = 0}^u \binom{k}{r}\binom{n-k}{p-r} \binom{n-p}{q-p} \binom{p-r}{2p-q-u} \left( \frac{1}{2} \right)^{\binom{p}{2} - \binom{r}{2} + (q-p)(2p-q-u)} \\
&\le \exp\left[ (\ln2) (\log n)^2 \left( (1+\eps-\theta) - \frac{1}{2}(1+\eps-\theta)^2 -\theta \left( \frac{5}{6}\eps-2\theta \right) + o(1) \right) \right],
\end{align*}
where in the first equality, $\binom{k}{r}\binom{n-k}{p-r}$ counts the number of choices of $C$ for $r$ ranging from $0$ to $u$, $\left( 1/2 \right)^{\binom{p}{2} - \binom{r}{2}}$ is the probability $C$ being a clique, $\binom{n-p}{q-p}$ is the number of choices of $U$, $\binom{p-r}{2p-q-u}$ is for $W$, and finally $\left( 1/2 \right)^{(q-p)(2p-q-u)}$ is the probability of $U,W$ being fully connected.
The lemma then follows from the Markov's inequality
\[
|\Psi_q \cap \Omega_{p, \le u}| \le X \le n \,\E[X],
\]
and a union bound over the choices of $q$, $p$, and $u$. 
\end{proof}


We now present the proof of \cref{thm:large-clique}.


\begin{proof}[Proof of \cref{thm:large-clique}]
By \cref{lem:rg_property1,lem:gateway}, the random graph $\GG(n,\frac{1}{2},\floor{n^\alpha})$ satisfies both the clique-counts properties $\pupp(\eps_0)$ and $\plow(\eps_0)$ for $\eps_0 = \alpha \le 1-\eps$ and the gateway-counts properties $\pgw$ simultaneously with probability $1-o(1)$ as $n \to \infty$. 
Throughout the proof we assume that $\pupp(\eps_0)$, $\plow(\eps_0)$, and $\pgw$ are all satisfied. 

Suppose $\hat{\beta}$ is such that $\hat{\beta}=\beta/ ((\ln 2) ( \log n))$ so that $\hat{\beta} = O(1)$. 
Pick a constant $\theta \in (0,\eps/3)$ such that
\[
\hat{\beta} \theta \le \frac{1}{2} \left( (1-\alpha)\gamma - \frac{\gamma^2}{2} \right). 
\]
Let $q = (1+\eps) \log n$, $p = (1+\eps-\theta) \log n$, and $r = \gamma \log n$. 
We define
\[
\mathcal{B} = \left( \Psi_{q} \cap \Omega_{p,<r} \right) \cup \Omega_{<q,r}
\]
to be the ``bottleneck'' set to which we will apply \cref{lem:conductance}.
Let $\mathcal{A} \subseteq \XX$ denote the collection of cliques that are reachable from the empty clique through a path (i.e.\ a sequence of cliques where each adjacent pair differs by exactly one vertex) not including any clique from $\mathcal{B}$ except possibly for the destination.
The following claim, whose proof is postponed to the end of this subsection, follows easily from the definitions of $\mathcal{A}$ and $\mathcal{B}$.  
\begin{claim}\label{claim:AB}
\begin{enumerate}
    \item Cliques from $\mathcal{A} \setminus \mathcal{B}$ are not adjacent to cliques from $\mathcal{A}^\ccomp$ (i.e., they differ by at least two vertices);
    \item $\Omega_{q,*} \subseteq \mathcal{A}^\ccomp \setminus \mathcal{B}$;
    \item $\Omega_{*,r} \subseteq \mathcal{A}^\ccomp \cup \mathcal{B}$.
\end{enumerate}
\end{claim}

Now observe from \cref{claim:AB} that to prove what we want, it suffices to show that starting from an appropriate state the Metropolis process does not hit any clique from $\mathcal{B}$ in $\exp(c \log^2 n)$-time with probability $1 - \exp(-c\log^2 n)$.
For a collection $\mathcal{U} \subseteq \XX$ of cliques we write $Z(\mathcal{U}) = Z(\beta;\mathcal{U}) = \sum_{\sigma \in \mathcal{U}} e^{\beta h_{|\sigma|}}$ to represent the partition function restricted to the set $\mathcal{U}$.
Since the boundary of $\mathcal{A}$ is included in $\mathcal{B}$ by \cref{claim:AB} it suffices to show that there exists a constant $c>0$ such that
\begin{align}\label{eq:large-goal_cond}
  \frac{ Z(\mathcal{B}) }{ Z(\mathcal{A}) } 
  \le \exp\left(-c \log^2 n \right).
\end{align}
Given \eqref{eq:large-goal_cond}, \cref{thm:large-clique} is an immediate consequence of \cref{lem:conductance}.

Observe that we have the following inclusion,
\[
\mathcal{A} \supseteq \Omega_{p,<r} \cup \Omega_{\le p,0} .
\]
To see this, for every clique in $\Omega_{p,<r}$, it can be reached from the empty clique by adding vertices one by one in any order, so that none of the intermediate cliques are from $\Omega_{p,<r} \supseteq \left( \Psi_{q} \cap \Omega_{p,<r} \right)$ or from $\Omega_{<q, r}$, except for possibly the last one. Similarly, every clique in $\Omega_{\le p,0}$ can be reached from the empty clique in the same way. (Note that cliques in $\Omega_{\le q,0}$, however, may not be reachable from $\emptyset$ without crossing $\mathcal{B}$; for example, by adding vertices one by one to reach a clique of size $q$, it may first reach a clique of size $p<q$ which is a $q$-gateway with intersection $<r$.)
Thus, we have
\begin{equation}\label{eq:B/A}
\frac{ Z(\mathcal{B}) }{ Z(\mathcal{A}) } 
\le \frac{ Z\left( \Psi_{q} \cap \Omega_{p,<r} \right) }{ Z(\mathcal{A}) } + \frac{ Z(\Omega_{<q, r}) }{ Z(\mathcal{A}) } 
\le \frac{ Z\left( \Psi_{q} \cap \Omega_{p,<r} \right) }{ Z(\Omega_{p,<r}) } + \frac{ Z(\Omega_{<q, r}) }{ Z(\Omega_{\le p,0}) } 
= \frac{ \left| \Psi_{q} \cap \Omega_{p,<r} \right| }{ W_{p,<r} } + \frac{ Z_{<q, r} }{ Z_{\le p,0} }. 
\end{equation}
The rest of the proof aims to upper bound the two ratios in \cref{eq:B/A} respectively.

For the first ratio, the key observation is that since $\gamma \le 1-\alpha$, $\XX_{p,<r}$ is dominated by cliques of intersection $o(\log n)$ (almost completely outside the planted clique) or more rigorously speaking those with sufficiently small intersection, say, in $\XX_{p,\le u}$ where $u = (\eps/6) \log n$. Hence, we write
\[
\left| \Psi_{q} \cap \Omega_{p,<r} \right|
=
\left| \Psi_{q} \cap \Omega_{p,\le u} \right|
+
\left| \Psi_{q} \cap \Omega_{p,u < \cdot < r} \right|
\le 
\left| \Psi_{q} \cap \Omega_{p,\le u} \right|
+
W_{p,u < \cdot < r},
\]
and combining $W_{p,<r} \ge W_{p,0}$ we get 
\begin{equation}\label{eq:bottle-split}
\frac{ \left| \Psi_{q} \cap \Omega_{p,<r} \right| }{ W_{p,<r} } 
\le \frac{ \left| \Psi_{q} \cap \Omega_{p,\le u} \right| }{ W_{p,0} }
+ \frac{ W_{p,u < \cdot < r} }{ W_{p,0} }. 
\end{equation}
By \cref{basic_rg}, the clique-counts property $\plow(\eps_0)$, and the gateway-counts property $\pgw$, we upper bound the first term in \cref{eq:bottle-split} by
\begin{align}\label{eq:bound-gateway}
\frac{ \left| \Psi_{q} \cap \Omega_{p,\le u} \right| }{ W_{p,0} }
&\le \frac{ 2 \left| \Psi_{q} \cap \Omega_{p,\le u} \right| }{ \E[W_{p,0}] } \nonumber\\
&\le \exp\left[ (\ln2) (\log n)^2 \left( (1+\eps-\theta) - \frac{1}{2}(1+\eps-\theta)^2 -\theta \left( \frac{5}{6}\eps-2\theta \right) + o(1) \right) \right] \nonumber\\
&~~\cdot \exp\left[ (\ln2) (\log n)^2 \left( - (1+\eps-\theta) + \frac{1}{2}(1+\eps-\theta)^2 + o(1) \right) \right] 
\nonumber\\
&= \exp\left[ (\ln2) (\log n)^2 \left( -\theta \left( \frac{5}{6}\eps-2\theta \right) + o(1) \right) \right].
\end{align}
For the second term in \cref{eq:bottle-split}, $\pupp(\eps_0)$ and $\plow(\eps_0)$ imply that 
\begin{equation}\label{eq:bound-ratio-W}
\frac{ W_{p,u < \cdot < r} }{ W_{p,0} }
\le 2n^3 \,\frac{\E \left[W_{p,u < \cdot < r} \right]}{ \E \left[W_{p,0}\right] }
\le 2n^4 \,\frac{\E \left[W_{p,u} \right]}{ \E \left[W_{p,0}\right] }
\le \exp\left[ (\ln2) (\log n)^2 \left( -\frac{1}{12}(1-\alpha)\eps + o(1) \right) \right],
\end{equation}
where the last inequality uses $\eps \le 1-\alpha$.
Combining \cref{eq:bottle-split,eq:bound-gateway,eq:bound-ratio-W}, we obtain 
\[
\frac{ \left| \Psi_{q} \cap \Omega_{p,<r} \right| }{ W_{p,<r} }
\le 
\exp\left( -c_1 \log^2 n + o(\log^2 n) \right)
\]
for some constant $c_1 = c_1(\alpha,\eps,\theta)$. This bounds the first ratio in \cref{eq:B/A}. 

For the second ratio in \cref{eq:B/A}, we have from $\pupp(\eps_0)$ and $\plow(\eps_0)$ that 
\begin{equation}\label{eq:second-ratio}
\frac{ Z_{<q, r} }{ Z_{\le p,0} }
\le 2n^3\, \frac{\E\left[ Z_{<q, r} ) \right]}{ \E\left[ Z_{\le p,0} \right] }.
\end{equation}
Using linearity of expectation we have that
\[
\E\left[ Z_{<q, r} \right]
= \sum_{q' = r}^{q-1} \E\left[ Z_{q',r} \right]
= \sum_{q' = r}^{q-1} \E\left[ W_{q',r} \right] \exp\left( \beta\ham_{q'} \right).
\]
Let 
\[
q^* = \argmax_{q' \in [n]: \; r \le q' < q}\, \E\left[ W_{q',r} \right] \exp\left( \beta\ham_{q'} \right).
\]
It follows that 
\begin{equation}\label{eq:222Z*r}
\E\left[ Z_{<q, r} \right] \le n \,\E\left[ W_{q^*,r} \right] \exp\left( \beta\ham_{q^*} \right).
\end{equation}
Meanwhile, let $q^{*\prime} = \min\{q^*, p\}$, and we have
\begin{equation}\label{eq:222Z*ler}
\E\left[ Z_{\le p,0} \right] \ge \E\left[ W_{q^{*\prime},0} \right] \exp\left( \beta\ham_{q^{*\prime}} \right).
\end{equation}
Combining \cref{eq:second-ratio,eq:222Z*r,eq:222Z*ler}, we get that
\[
\frac{ Z_{<q, r} }{ Z_{\le p,0} } \le 2n^4 \,\frac{\E\left[ W_{q^*,r} \right]}{\E\left[ W_{q^{*\prime},0} \right]} \exp\left[ \beta(h_{q^*} - h_{q^{*\prime}}) \right].
\]
Let $\rho=q^*/ \log n$ and $\rho'=q^{*\prime}/ \log n$.
Then by definition we have that $\rho' \le \rho \le 1+\eps$ and $\rho' = \min\{\rho, 1+\eps-\theta\}$.
Furthermore by \cref{ass:ham} we have
\[
\beta(h_{q^*} - h_{q^{*\prime}}) \le (\ln2) (\log n)^2 \hat{\beta} (\rho - \rho').
\]
Then, an application of \cref{item:X_qr-expectation} of \cref{basic_rg} implies that 
\begin{align*}
\frac{ Z_{<q, r} }{ Z_{\le p,0} } 
&\le \exp\left[ (\ln2) (\log n)^2 \left( \rho - \frac{\rho^2}{2} - (1-\alpha)\gamma + \frac{\gamma^2}{2} - \rho' + \frac{(\rho')^2}{2} + \hat{\beta} (\rho - \rho') + o(1) \right) \right]\\
&= \exp\left[ (\ln2) (\log n)^2 \left( (\rho-\rho') \left( \hat{\beta} + 1 - \frac{1}{2}(\rho+\rho') \right) - \left( (1-\alpha)\gamma - \frac{\gamma^2}{2} \right) + o(1) \right) \right].
\end{align*}
If $\rho = \rho'$ then $Z_{<q, r} / Z_{\le p,0} \le \exp\left( -c_2 \log^2 n + o(\log^2 n) \right)$ for $c_2 = (\ln 2) \left( (1-\alpha)\gamma - \gamma^2/2 \right)$.
If $\rho > \rho'$ then 
\[
(\rho-\rho') \left( \hat{\beta} + 1 - \frac{1}{2}(\rho+\rho') \right)
\le \theta \hat{\beta} \le 
\frac{1}{2} \left( (1-\alpha)\gamma - \frac{\gamma^2}{2} \right)
\]
since $1 \le 1+\eps-\theta =\rho' < \rho \le 1+\eps$.
Therefore, we have $Z_{<q, r} / Z_{\le p,0}  \le \exp\left( -c_2 \log^2 n + o(\log^2 n) \right)$ for $c_2 = \frac{\ln 2}{2}\left( (1-\alpha)\gamma - \gamma^2/2 \right)$.
This bounds the second ratio in \cref{eq:B/A}. 
Hence, we establish \cref{eq:large-goal_cond} and the theorem then follows.
\end{proof}

\begin{proof}[Proof of \cref{claim:AB}]
The first item is obvious, since if a clique $\sigma \in \mathcal{A} \setminus \mathcal{B}$ is adjacent to another clique $\sigma'$, then by appending $(\sigma,\sigma')$ to the path from $\emptyset$ to $\sigma$ we get a path from $\emptyset$ to $\sigma'$ without passing through cliques from $\mathcal{B}$ except possibly at $\sigma''$, since $\sigma \notin \mathcal{B}$. This implies that $\sigma'' \in \mathcal{A}$ which proves the first item. 

For the second item, suppose for contradiction that there exists a clique $C$ of size $q$ that is in $\mathcal{A} \cup \mathcal{B}$. 
Clearly $C \notin \mathcal{B}$ and so $C \in \mathcal{A}$.
Then there exists a path of cliques $\emptyset = C_0, C_1, \dots, C_\ell = C$ which contain no cliques from $\mathcal{B}$. 
Let $C_j$ for $j \in [\ell]$ be the first clique of size $q$ in this path; that is, $|C_i| < q$ for $0 \le i < j$. 
Then, the (sub)path $\emptyset = C_0, C_1, \dots, C_j$ must contain a $q$-gateway of size $p$, call it $C'$, as one can choose the largest $i<j$ for which $|C_i|=p$ and set $C'=C_i$. 
If $C'$ has intersection with the planted clique less than $r$, then $C' \in \Psi_{q} \cap \Omega_{p,<r} \subseteq \mathcal{B}$, contradiction. 
Otherwise, $C'$ has intersection at least $r$ which means at some earlier time, it will pass a clique of intersection exactly $r$ whose size is less than $q$, which is again in $\Omega_{<q,r} \subseteq \mathcal{B}$ leading to a contradiction. This establishes the second item.

For the third one, if $\sigma$ is a clique of intersection $r$, then either $|\sigma| < q$ meaning $\sigma \in \Omega_{<q,r} \subseteq \mathcal{B}$ or $|\sigma| \ge q$ meaning
any path from $\emptyset$ to $\sigma$ must pass through a clique of size exactly $q$ and by the second item must pass through cliques from $\mathcal{B}$, implying $\sigma \in \mathcal{A}^\ccomp$.
\end{proof}

\subsection{Starting from the Empty Clique}

\begin{theorem}\label{thm:start-from-empty-hit-large}
Let $\alpha \in [0,1)$ be any fixed constant.
For any constant $\eps \in (0,1)$, the random graph $\GG(n,\frac{1}{2},k = \floor{n^\alpha})$ with a planted clique satisfies the following with probability $1-o(1)$ as $n\to \infty$. 

Consider the general Gibbs measure given by \cref{eq:Gibbs-general} for arbitrary $1$-Lipschitz $\ham$ with $h_0 = 0$ and inverse temperature $\beta = o(\log n)$.
Let $\{X_t\}$ denote the Metropolis process on $G$ with stationary distribution $\pi_\beta$. 
Then there exist constants $\rhozero = \rhozero(\alpha,\eps)>0$ and $c = c(\alpha,\eps)>0$ such that for any clique $C \in \XX$ of size at most $\rhozero\log n$, one has 
\[
\Pr\Big(\exists t \in \N \wedge t \le n^{c \log n} \text{~s.t.~} |X_t| \ge (1+\eps) \log n \;\Big\vert\; X_0 = C \Big) \le  n^{-c \log n}.
\]
In particular, the Metropolis process starting from the empty clique requires $n^{\Omega(\log n)}$ steps to reach cliques of size at least $(1+\eps)\log n$, with probability $1 - n^{-\Omega(\log n)}$. 
\end{theorem}

\begin{proof}
By \cref{lem:rg_property1,lem:satisfy-cond}, the random graph $\GG(n,\frac{1}{2},\floor{n^\alpha})$ satisfies both $\pupp(0.1)$ and $\pexp(\eta)$ for $\eta = \eps/2$ with probability $1-o(1)$ as $n\to \infty$.
In the rest of the proof we assume that both $\pupp(0.1)$ and $\pexp(\eta)$ are satisfied.

Suppose that $|C| = \qzero = \rhozero' \log n$, for $0 \leq \rhozero' \leq \rhozero$ and $\rhozero>0$ is a sufficiently small constant to be determined. 
Let $q = (1+\eps) \log n$ and $s = (1-\alpha)\log n$.
First observe that if the chain arrives at a clique in $\Omega_{q,*}$, it either hits a clique from $\Omega_{q,<s}$ or previously reached a clique in $\Omega_{<q,s}$.
This implies that
\begin{multline*}
\Pr\Big( \exists t \in \N \wedge t \le T:\, X_t \in \XX_{q,*} \;\Big\vert\; X_0 = C \Big) \\
\le
\Pr\Big( \exists t \in \N \wedge t \le T:\, X_t \in \XX_{q,<s} \;\Big\vert\; X_0 = C \Big) 
+ 
\Pr\Big( \exists t \in \N \wedge t \le T:\, X_t \in \XX_{<q,s} \;\Big\vert\; X_0 = C \Big).  
\end{multline*}
It suffices to upper bound each of the two terms respectively.

Similar as in \cref{eq:final_ub}, we deduce from the union bound and $\pupp(0.1)$ that
\begin{align*}
\Pr\Big( \exists t \in \N \wedge t \le T:\, X_t \in \XX_{q,<s} \;\Big\vert\; X_0 = C \Big) 
&\le Tn \,\max_{t \in [T]} \,\max_{r \in [s]} \Pr\left( X_t \in \XX_{q,r} \mid X_0 = C \right); \\
\Pr\Big( \exists t \in \N \wedge t \le T:\, X_t \in \XX_{<q,s} \;\Big\vert\; X_0 = C \Big) 
&\le Tn \,\max_{t \in [T]} \,\max_{p \in [q]} \Pr\left( X_t \in \XX_{p,s} \mid X_0 = C \right).
\end{align*} 
It suffices to show that for all integer $t \ge 1$, all pairs 
\[
(p,r) \in \left\{(q,r): r \in [s]\right\} \cup \left\{(p,s): p \in [q]\right\},
\]
and all clique $\sigma \in \XX_{p,r}$, it holds
\begin{equation}\label{eq:goal-1+eps}
\Pr\left( X_t \in \XX_{p,r} \mid X_0 = C \right) \le t n^{-c\log n} 
\end{equation}
for some constant $c>0$.

Without loss of generality we may assume that $\eps \le 1-\alpha$.
Let $\eta = \eps/2$ and $\rhozero = \eps^2/8$.
Write $\beta = (\ln 2) \hat{\beta} \log n$ with $\hat{\beta} = o(1)$.
Consider first the pair $(q,r)$ where $r\in[s]$. Suppose that $r = \gamma \log n$ with $0\le \gamma \le 1-\alpha$.
We deduce from \cref{lem:reaching-prob}, with $\hat{\beta} = o(1)$, $\rho = 1+\eps$, $\gamma \in [0,1-\alpha]$, and $\rho' = \min\{\rho, 1-\eta\} = 1-\eta$, that
\begin{align*}
& \Pr\left( X_t \in \XX_{q,r} \mid X_0 = C \right) \\
\le{}&
\exp\left[ (\ln 2)(\log n)^2 \left( \left( (1+\eps) - \frac{1}{2}(1+\eps)^2 \right) - \left( (1-\eta) - \frac{1}{2}(1-\eta)^2 \right) - \left( (1-\alpha)\gamma - \frac{\gamma^2}{2} - \rhozero + \frac{\rhozero^2}{2} \right) + o(1) \right) \right] \\
\le{}& 
\exp\left[ (\ln 2)(\log n)^2 \left( - \frac{\eps^2}{2} + \frac{\eta^2}{2} + \rhozero + o(1) \right) \right]
= \exp\left[ (\ln 2)(\log n)^2 \left( - \frac{\eps^2}{4} + o(1) \right) \right].
\end{align*}
This shows \cref{eq:goal-1+eps} for the first case. 

For the second case, we have the pair $(p,s)$ where $p\in [q]$. Suppose $p = \rho \log n$ with $0\le \rho \le 1+\eps$. Also recall that $s=(1-\alpha)\log n$. Again, we deduce from \cref{lem:reaching-prob}, with $\hat{\beta} = o(1)$, $\rho \in [0,1+\eps]$, $\gamma = 1-\alpha$, and $\rho' = \min\{\rho, 1-\eta\}$, that
\begin{align*}
& \Pr\left( X_t \in \XX_{p,s} \mid X_0 = C \right) \\
\le{}& 
\exp\left[ (\ln 2)(\log n)^2 \left( \left( \rho - \frac{\rho^2}{2} \right) - \left( \rho' - \frac{(\rho')^2}{2} \right) - \left( \frac{1}{2}(1-\alpha)^2 - \rhozero + \frac{\rhozero^2}{2} \right) + o(1) \right) \right] \\
\le{}& \exp\left[ (\ln 2)(\log n)^2 \left( - \frac{1}{2}(\rho-1)^2 + \frac{1}{2}(\rho'-1)^2 - \frac{3}{8}(1-\alpha)^2 + o(1) \right) \right], 
\end{align*}
where the last inequality follows from
\[
\frac{1}{2}(1-\alpha)^2 - \rhozero + \frac{\rhozero^2}{2}
\ge \frac{1}{2}(1-\alpha)^2 - \frac{\eps^2}{8}
\ge \frac{3}{8}(1-\alpha)^2
\]
since $\eps \le 1-\alpha$.
If $\rho' = \rho \le 1-\eta$, then $- \frac{1}{2}(\rho-1)^2 + \frac{1}{2}(\rho'-1)^2 = 0$. 
If $\rho' = 1-\eta < \rho \le 1+\eps$, then
\[
- \frac{1}{2}(\rho-1)^2 + \frac{1}{2}(\rho'-1)^2
\le \frac{\eta^2}{2} \le \frac{\eps^2}{8} 
\le \frac{1}{8}(1-\alpha)^2
\]
where the last inequality is because $\eps \le 1-\alpha$.
Hence,
\[
\Pr\left( X_t \in \XX_{p,s} \mid X_0 = C \right)
\le \exp\left[ (\ln 2)(\log n)^2 \left( - \frac{1}{4}(1-\alpha)^2 + o(1) \right) \right].
\]
This shows \cref{eq:goal-1+eps} for the second case. 
The theorem then follows from \cref{eq:goal-1+eps}.
\end{proof}

\subsection{Low Temperature Regime and Greedy Algorithm}

\begin{theorem}
\label{thm:greedy}
Let $\alpha \in [0,1)$ be any fixed constant.
For any constant $\eps \in (0,1)$, the random graph $\GG(n,\frac{1}{2},k = \floor{n^\alpha})$ with a planted clique satisfies the following with probability $1-o(1)$ as $n\to \infty$. 

Consider the general Gibbs measure given by \cref{eq:Gibbs-general} for the identity function $\ham_q = q$ and inverse temperature $\beta = \omega(\log n)$.
For any constant $\gamma \in (0,1-\alpha]$, there exists constant $\rhozero = \rhozero(\alpha,\eps)>0$ such that for any clique $C \in \XX$ of size at most $\rhozero\log n$ and any constant $c \in \N^+$, 
with probability at least $1-n^{-\omega(1)}$ the Metropolis process starting from $C$ will not reach  
\begin{itemize}
    \item either cliques of size at least $(1+\eps)\log n$, 
    \item or cliques of intersection at least $\gamma \log n$ with the planted clique,
\end{itemize}
within $n^c$ steps.
\end{theorem}

For a subset $\mathcal{S} \subseteq \XX$ of cliques and an integer $p \in \N^+$, let $\mathcal{S}^{[p]}$ denote the collection of all cliques of size $p$ that are subsets of cliques from $\mathcal{S}$, i.e.,
\[
\mathcal{S}^{[p]} = \left\{ C \in \XX_{p,*}: \exists \sigma \in \mathcal{S} \mathrm{~s.t.~} C \subseteq \sigma \right\}.
\]
For $0\le r \le q$ we define $W^{[p]}_{q,r} = \left| \Omega^{[p]}_{q,r} \right|$ and similarly for $W^{[p]}_{q,<r}$, $W^{[p]}_{<q,r}$, etc.
\begin{lemma}\label{lem:extend}
Consider the random graph $\GG(n,\frac{1}{2},k = \floor{n^\alpha})$ with a planted clique conditional on satisfying the property $\pupp(0.1)$. 
For any $0\le p \le q = \rho \log n$ with $\rho \ge 0$ and any $r = \gamma \log n$ with $0\le \gamma \le \rho$ we have
\[
W^{[p]}_{q,r} \le \exp\left[ (\ln2) (\log n)^2 \left( \rho - \frac{\rho^2}{2} - (1-\alpha)\gamma + \frac{\gamma^2}{2} + o(1) \right) \right]
\]
with high probability as $n \to \infty$.
\end{lemma}
\begin{proof}
Notice that every clique of size $q$ has $\binom{q}{p} \le 2^q \le n^2$ subsets of size $p$. The lemma then follows immediately from $\pupp(\eps)$.
\end{proof}

We now give the proof of \cref{thm:greedy}.

\begin{proof}[Proof of \cref{thm:greedy}]
By \cref{lem:rg_property1,lem:satisfy-cond}, the random graph $\GG(n,\frac{1}{2},\floor{n^\alpha})$ satisfies both $\pupp(0.1)$ and $\pexp(\eta)$ for $\eta = \eps/2$ with probability $1-o(1)$ as $n\to \infty$.
In the rest of the proof we assume that both $\pupp(0.1)$ and $\pexp(\eta)$ are satisfied.

First, a simple observation is that the process actually never remove vertices. Indeed, the probability of removing a vertex from the current clique in one step is at most $e^{-\beta} = n^{-\omega(1)}$. Since we run the Metropolis process for polynomially many steps, the probability that the chain ever remove a vertex is upper bounded by $\poly(n) \cdot n^{-\omega(1)} =  n^{-\omega(1)}$. Hence, in this low temperature regime the Metropolis process is equivalent to the greedy algorithm.

Without loss of generality, we may assume that
\[
\eps \le \sqrt{ 2\left((1-\alpha)\gamma - \frac{\gamma^2}{2} \right)}.
\]
Let $q = (1+\eps) \log n$ and $s = \gamma\log n$.
Suppose the initial state is a clique $C$ of size $\rhozero \log n$ for $\rhozero \le \eps^2/8$.
Let $\eta = \eps/2$
and $q' = (1-\eta) \log n$.
If the chain arrives at a clique in $\Omega_{q,*}$, it either hits a clique from $\Omega_{q,<s}$ or previously reached a clique in $\Omega_{<q,s}$. 
Similarly, if the chain hits $\Omega_{*,s}$, then it must also reach either $\Omega_{<q,s}$ or $\Omega_{q,<s}$.
Hence, to bound the probability of reaching either $\Omega_{q,*}$ or $\Omega_{*,s}$, we only need to bound the probability of reaching $\Omega_{q,<s}$ or $\Omega_{<q,s}$.
Furthermore, since the process never removes a vertex with high probability in polynomially many steps, in the case this happens it must first reach a clique from $\Omega^{[q']}_{q,<s}$, or $\Omega^{[q']}_{q' < \cdot <q,s}$, or $\Omega_{\le q',s}$.
To summarize, we can deduce from the union bound that
\begin{align*}
&\Pr\Big( \exists t \in \N \wedge t \le T:\, X_t \in \XX_{q,*} \cup \XX_{*,s} \;\Big\vert\; X_0 = C \Big) \\
\le{}&
\Pr\Big( \exists t \in \N \wedge t \le T:\, X_t \in \XX^{[q']}_{q,<s} \;\Big\vert\; X_0 = C \Big) 
+ 
\Pr\Big( \exists t \in \N \wedge t \le T:\, X_t \in \XX^{[q']}_{q' < \cdot <q,s} \;\Big\vert\; X_0 = C \Big) \\
& +
\Pr\Big( \exists t \in \N \wedge t \le T:\, X_t \in \XX_{\le q',s} \;\Big\vert\; X_0 = C \Big) + \frac{1}{n^{\omega(1)}}. \\
\end{align*}
We bound each of the three probabilities respectively.

For the first case, we have from the union bound and $\pupp(0.1)$ that,
\[
\Pr\Big( \exists t \in \N \wedge t \le T:\, X_t \in \XX^{[q']}_{q,<s} \;\Big\vert\; X_0 = C \Big) 
\le
Tn^4 \,\max_{t \in [T]} \,\max_{r \in [s]} \max_{\sigma \in \XX^{[q']}_{q,r}} \E\left[ W^{[q']}_{q,r} \right] \Pr\left( X_t = \sigma \mid X_0 = C \right).
\]
Suppose $r = \gamma' \log n \in [s]$. 
We deduce from \cref{basic_rg,lem:extend} and \cref{lem:reaching-prob} (with $\rho = \rho' = 1-\eta$ and $\gamma = 0$ for the notations of \cref{lem:reaching-prob}) that for every integer $t\ge 1$ and every clique $\sigma \in \XX^{[q']}_{q,r}$ (note that $|\sigma| = q'$ since $q' < q$),
\begin{align*}
& \E\left[ W^{[q']}_{q,r} \right] \Pr\left( X_t = \sigma \mid X_0 = C \right) \\
\le{}& \frac{\E\left[ W^{[q']}_{q,r} \right]}{\E\left[ W_{q',0} \right]} \cdot \E\left[ W_{q',0} \right]  \Pr\left( X_t = \sigma \mid X_0 = C \right) \\
\le{}&
t \exp\left[ (\ln 2)(\log n)^2 \left( (1+\eps) - \frac{1}{2}(1+\eps)^2 - (1-\alpha)\gamma' + \frac{(\gamma')^2}{2} - (1-\eta) + \frac{1}{2}(1-\eta)^2 + \rhozero - \frac{\rhozero^2}{2} + o(1) \right) \right] \\
\le{}& 
t \exp\left[ (\ln 2)(\log n)^2 \left( - \frac{\eps^2}{2} + \frac{\eta^2}{2} + \rhozero + o(1) \right) \right]
\le 
t \exp\left[ (\ln 2)(\log n)^2 \left( - \frac{\eps^2}{4} + o(1) \right) \right],
\end{align*}
where the last inequality follows from
\[
\frac{\eta^2}{2} + \rhozero \le \frac{\eps^2}{8} + \frac{\eps^2}{8} = \frac{\eps^2}{4}.
\]

Next consider the second case, where we have 
\[
\Pr\Big( \exists t \in \N \wedge t \le T:\, X_t \in \XX^{[q']}_{q' < \cdot <q,s} \;\Big\vert\; X_0 = C \Big)
\le
Tn^4 \,\max_{t \in [T]} \,\max_{p \in [q]\setminus [q']} \max_{\sigma \in \XX^{[q']}_{p,s}} \E\left[ W^{[q']}_{p,s} \right] \Pr\left( X_t = \sigma \mid X_0 = C \right).
\]
Suppose $p = \rho \log n \in [q] \setminus [q']$ and so $1-\eta \le \rho \le 1+\eps$. Also recall that $s=\gamma\log n$.
We deduce from \cref{basic_rg,lem:extend} and \cref{lem:reaching-prob} (with $\rho = \rho' = 1-\eta$ and $\gamma = 0$ for the notations of \cref{lem:reaching-prob}) that for every integer $t\ge 1$ and every clique $\sigma \in \XX^{[q']}_{p,s}$ (note that $|\sigma| = q'$ since $q' < q$),
\begin{align*}
& \E\left[ W^{[q']}_{p,s} \right] \Pr\left( X_t = \sigma \mid X_0 = C \right) \\
\le{}& 
\frac{\E\left[ W^{[q']}_{p,s} \right]}{\E[W_{q',0}]} \cdot \E[W_{q',0}] \Pr\left( X_t = \sigma \mid X_0 = C \right) \\
\le{}& 
t \exp\left[ (\ln 2)(\log n)^2 \left( \rho - \frac{\rho^2}{2} - (1-\alpha)\gamma + \frac{\gamma^2}{2} - (1-\eta) + \frac{1}{2}(1-\eta)^2 + \rhozero - \frac{\rhozero^2}{2} + o(1) \right) \right] \\
\le{}& 
t \exp\left[ (\ln 2)(\log n)^2 \left( - \left( (1-\alpha)\gamma - \frac{\gamma^2}{2} \right) + \frac{\eta^2}{2} + \rhozero + o(1) \right) \right] \\
\le{}& 
t \exp\left[ (\ln 2)(\log n)^2 \left( - \frac{1}{2} \left( (1-\alpha)\gamma - \frac{\gamma^2}{2} \right) + o(1) \right) \right],
\end{align*}
where the last inequality follows from
\[
\frac{\eta^2}{2} + \rhozero 
\le \frac{\eps^2}{4}
\le \frac{1}{2} \left( (1-\alpha)\gamma - \frac{\gamma^2}{2} \right).
\]

Finally consider the third case.
Again we have
\[
\Pr\Big( \exists t \in \N \wedge t \le T:\, X_t \in \XX_{\le q',s} \;\Big\vert\; X_0 = C \Big)
\le
Tn^4 \,\max_{t \in [T]} \,\max_{p \in [q']} \max_{\sigma \in \XX_{p,s}} \E[W_{p,s}] \Pr\left( X_t = \sigma \mid X_0 = C \right).    
\]
Suppose $p = \rho \log n \in [q']$ and so $\rho \le 1-\eta$. Also recall that $s=\gamma\log n$.
We deduce from \cref{lem:reaching-prob} that for every integer $t\ge 1$ and every clique $\sigma \in \XX_{p,s}$,

\begin{align*}
\E[W_{p,s}] \Pr\left( X_t = \sigma \mid X_0 = C \right) 
&\le 
t \exp\left[ (\ln 2)(\log n)^2 \left( - \left( (1-\alpha)\gamma - \frac{\gamma^2}{2} \right) + \rhozero - \frac{\rhozero^2}{2} + o(1) \right) \right] \\
&\le 
t \exp\left[ (\ln 2)(\log n)^2 \left( - \frac{3}{4} \left( (1-\alpha)\gamma - \frac{\gamma^2}{2} \right) + o(1) \right) \right],
\end{align*}
where the last inequality follows from
\[
\rhozero \le \frac{\eps^2}{8} \le \frac{1}{4}\left( (1-\alpha)\gamma - \frac{\gamma^2}{2} \right).
\]

Therefore, we conclude that
\[
\Pr\Big( \exists t \in \N \wedge t \le T:\, X_t \in \XX_{q,*} \cup \XX_{*,s} \;\Big\vert\; X_0 = C \Big)
\le \frac{Tn^2}{n^{\Omega(\log n)}} + \frac{1}{n^{\omega(1)}}
\le \frac{1}{n^{\omega(1)}},
\]
as we wanted.
\end{proof}

\section{Simulated Tempering}\label{sec:st}

In this section, we discuss our lower bounds against the simulated tempering versions of the Metropolis process.

\subsection{Definition of the dynamics}\label{sec:st_dfn}

We start with the formal definition.
Suppose for some $m \in \mathbb{N}$ we have a collection of inverse temperatures $\beta_0 < \beta_1 < \dots < \beta_m$. 
For each $i \in [m]$, let $\widehat{Z}(\beta_i)$ denote an estimate of the partition function $Z(\beta_i)$.
The simulated tempering (ST) dynamics is a Markov chain on the state space $\XX \times [m]$, which seeks to optimize a Hamiltonian defined on $\XX \times [m]$, say given according to $H(C)=\ham_{|C|}$ for an arbitrary vector $\{\ham_q, q \in [n]\}$. 
The transition matrix is given as follows.

\begin{itemize}
\item A level move: For $C,C' \in \XX$ and $i \in [m]$ such that $C$ and $C'$ differ by exactly one vertex,
\[
P_{\simt}((C,i), (C',i)) 
= \frac{\pa}{n} \min\left\{ 1, \exp\left[ \beta_i \left( \ham_{|C'|} - \ham_{|C|} \right) \right] \right\}
\]

\item A temperature move: For $i,i' \in [m]$ such that $|i-i'| = 1$,
\[
P_{\simt}((C,i), (C, i')) 
= \frac{1-\pa}{2} \min\left\{ 1, \frac{\widehat{Z}(\beta_i)}{\widehat{Z}(\beta_{i'})} \exp\left[ \left( \beta_{i'} - \beta_i \right) \ham_{|C|} \right] \right\}.
\]
\end{itemize}

Some remarks are in order.

\begin{remark}
The stationary distribution of the ST dynamics can be straightforwardly checked to be given by $\mu(C,i) \propto \frac{Z(\beta_i)}{\widehat{Z}(\beta_i)} \pi_{\beta_i}(C)$, for $\pi_{\beta_i}(C)$ the generic Gibbs measure defined in \cref{eq:Gibbs-general}. 
Notice that if $\widehat{Z}(\beta_i) = Z(\beta_i)$ for all $i \in [m]$ then we have
\[
\mu(C,i) = \frac{1}{m+1} \pi_{\beta_i}(C).
\]and along a single temperature the ST dynamics is identical to the Metropolis process introduced in \cref{sec:met_process}.
\end{remark}

\begin{remark}
The use of estimates $\hat{Z}(\beta_i)$ of the partition function $Z(\beta_i)$ in the definition of the ST dynamics, as opposed to the original values is naturally motivated from applications where one cannot efficiently compute the value of $Z(\beta_i)$ to decide the temperature move step. 
\end{remark}

\subsection{Existence of a Bad Initial Clique}
We now present our lower bound results which are for the ST dynamics under a worst-case initialization.

Our first result is about the ST dynamics failing to reach $\gamma \log n$ intersection with the planted clique, similar to the Metropolis process according to \cref{main_result_mixing}. Interestingly, the lower bound holds \emph{for any choice} of \emph{arbitrarily many} temperatures and \emph{for any choice} of estimators of the partition function.

\begin{theorem}\label{main_result_mixing_st}
Let $\alpha \in (0,1)$ be any fixed constant.
For any constant $\gamma > 0$, the random graph $\GG(n,\frac{1}{2},k = \floor{n^\alpha})$ with a planted clique satisfies the following with probability $1-o(1)$ as $n\to \infty$. 

Consider the general ST dynamics given in \cref{sec:st_dfn} for arbitrary $\ham,$ arbitrary $m \in \mathbb{N}$, arbitrary inverse temperatures $\beta_1<\beta_2<\ldots<\beta_m $, and arbitrary estimates $\hat{Z}(\beta_i), i=1,\ldots,m$. 
Then there is an initialization pair of temperature and clique for the ST dynamics from which it requires $\exp(\Omega(\log^2n))$-time to reach a pair of temperature and clique where the clique is of intersection with the planted clique at least $\gamma\log_2n,$ with probability at least $1-\exp(-\Omega(\log^2n))$. In particular, under worst-case initialization it fails to recover the planted clique in polynomial-time.
\end{theorem}

\begin{proof}
Throughout the proof we assume that both $\pupp(\eps)$ and $\plow(\eps)$ are satisfied for $\eps>0$ given by \cref{eq:def-eps-property}, which happens with probability $1-o(1)$ as $n\to \infty$ by \cref{lem:rg_property1}.

Notice that we can assume without loss of generality that $\gamma$ satisfies $0<\gamma < 2(1-\alpha)$. We let $r=\floor{\gamma\log n}$. Now from the proof of \cref{main_result_mixing} we have that for any such $\gamma,$ 
there exists a constant $c=c(\alpha,\gamma)>0$ such that for all $\beta_i, i=1,\ldots,m$ and the corresponding $\pi_{\beta_i}, i=1,\ldots,m$ Gibbs measure per \cref{eq:Gibbs-general},
\begin{align}\label{eq:goal_cond_st}
  \frac{ \pi_{\beta_i}(\XX_{*,r}) }{ \pi_{\beta_i}(\XX_{*,\le r}) } 
  = \frac{Z_{*,r}}{Z_{*,\le r}}
  \le \exp\left(-c \log^2 n \right),
\end{align}
w.h.p. as $n \rightarrow +\infty.$
We now consider the set $\mathcal{A}=\bigcup_{i \in [m]} \XX_{*, \leq r} \times \{\beta_i\} $ the subset of the state space of the ST dynamics, and notice $\partial \mathcal{A} = \bigcup_{i \in [m]} \XX_{*, r} \times \{\beta_i\},$ where $\partial A$ is the boundary of $A$.
In particular using \eqref{eq:goal_cond_st} we conclude that w.h.p. as $n \rightarrow +\infty$ 
\begin{align}\label{eq:goal_cond_st_st_2}
  \frac{ \mu(\partial \mathcal{A}) }{ \mu(\mathcal{A}) } & = \frac{ \sum_{i=1}^m \frac{Z(\beta_i)}{\widehat{Z}(\beta_i)}\pi_{\beta_i}(\XX_{*,r}) }{ \sum_{i=1}^m\frac{Z(\beta_i)}{\widehat{Z}(\beta_i)}\pi_{\beta_i}(\XX_{*,\le r}) }\le \exp\left(-c \log^2 n \right),
\end{align} Given \cref{eq:goal_cond_st_st_2}, \cref{main_result_mixing_st} is an immediate consequence of \cref{lem:conductance}.
\end{proof}

Our second result is about the ST dynamics under the additional restriction that $\max_{i \in [m]} |\beta_i|=O(\log n).$ In this case, similar to the Metropolis process per \cref{thm:large-clique}, we show that the ST dynamics fail to reach either $(1+\epsilon) \log n$-cliques or cliques with intersection at least $\gamma \log n$ with the planted clique. Interestingly, again, the lower bound holds \emph{for any choice} of \emph{arbitrarily many} temperatures of magnitude $O(\log n)$ and \emph{for any choice} of estimators of the partition function.

\begin{theorem}\label{thm:large-clique_st}
Let $\alpha \in [0,1)$ be any fixed constant.
Then the random graph $\GG(n,\frac{1}{2},k = \floor{n^\alpha})$ with a planted clique satisfies the following with probability $1-o(1)$ as $n\to \infty$. 

Consider the general ST dynamics given in \cref{sec:st_dfn} for arbitrary $\ham$ satisfying \cref{ass:ham}, arbitrary $m \in \mathbb{N}$, arbitrary inverse temperatures $\beta_1<\beta_2<\ldots<\beta_m $ with $\max_{i \in [m]}|\beta_i|=O(\log n)$, and arbitrary estimates $\hat{Z}(\beta_i), i=1,\ldots,m$. 
For any constants $\eps \in (0,1-\alpha)$ and $\gamma \in (0,1-\alpha]$, there is an initialization pair of temperature and clique for the ST dynamics from which it requires $\exp(\Omega(\log^2n))$-time to reach a pair of temperature and clique where
\begin{itemize}
    \item either the clique is of size at least $(1+\eps)\log n$, 
    \item or the clique is of intersection at least $\gamma \log n$ with the planted clique,
\end{itemize}
with probability at least $1-\exp(-\Omega(\log^2n))$.
\end{theorem} 

\begin{proof}
Throughout the proof we assume that $\pupp(\eps_0)$, $\plow(\eps_0)$, and $\pgw$ are all satisfied for $\eps_0 = \alpha \le 1-\eps$, which happens with probability $1-o(1)$ as $n\to \infty$ by \cref{lem:rg_property1,lem:gateway}. 

We start with following the proof of Theorem \ref{thm:large-clique}.
For $i=1,\ldots,m$ let $\hat{\beta}_i$ be such that $\hat{\beta}_i=\beta_i/ ((\ln 2) ( \log n))$. By assumption we have $\max_{i \in [m]}|\hat{\beta_i}| = O(1)$. 
Pick a constant $\theta \in (0,\eps/3)$ such that for all $i \in [m]$
\[
\hat{\beta}_i \theta \le \frac{1}{2} \left( (1-\alpha)\gamma - \frac{\gamma^2}{2} \right). 
\]
Let $q = (1+\eps) \log n$, $p = (1+\eps-\theta) \log n$, and $r = \gamma \log n$. 
Let also
\[
\mathcal{B} = \left( \Psi_{q} \cap \Omega_{p,<r} \right) \cup \Omega_{<q,r}.
\] 
Let also $\mathcal{A} \subseteq \XX$ denote the collection of cliques that are reachable from the empty clique through a path (i.e. a sequence of cliques where each adjacent pair differs by exactly one vertex) not including any clique from $\mathcal{B}$ except possibly for the destination.

From the proof of Theorem \ref{thm:large-clique} we have that there exists a constant $c=c(\alpha,\gamma,\theta)>0$ such that for all $\beta_i, i=1,\ldots,m$ and the corresponding $\pi_{\beta_i}, i=1,\ldots,m$ being the Gibbs measure per \cref{eq:Gibbs-general},
\begin{align}\label{eq:goal_cond_st_2}
  \frac{ \pi_{\beta_i}(\partial \mathcal{A}) }{ \pi_{\beta_i}(\mathcal{A}) } 
 \le \frac{ \pi_{\beta_i}(\mathcal{B})}{ \pi_{\beta_i}(\mathcal{A}) } 
  \le \exp\left(-c \log^2 n \right),
\end{align}
w.h.p. as $n \rightarrow +\infty.$
We now consider the set $\mathcal{G}=\bigcup_{i \in [m]} \mathcal{A} \times \{\beta_i\} $ the subset of the state space of the ST dynamics, and notice $\partial \mathcal{G} = \bigcup_{i \in [m]} \partial \mathcal{A} \times \{\beta_i\}.$
In particular using \eqref{eq:goal_cond_st_2} we conclude that w.h.p. as $n \rightarrow +\infty$ 
\begin{align}\label{eq:goal_cond_st_st_3}
  \frac{ \mu(\partial \mathcal{G}) }{ \mu(\mathcal{G}) } & = \frac{ \sum_{i=1}^m \frac{Z(\beta_i)}{\widehat{Z}(\beta_i)}\pi_{\beta_i}(\partial \mathcal{A}) }{ \sum_{i=1}^m\frac{Z(\beta_i)}{\widehat{Z}(\beta_i)}\pi_{\beta_i}(\mathcal{A}) }\le \exp\left(-c \log^2 n \right),
\end{align} Given \cref{eq:goal_cond_st_st_3}, \cref{main_result_mixing_st} is an immediate consequence of \cref{lem:conductance}.
\end{proof}

\subsection{Starting From the Empty Clique}


\begin{theorem}\label{thm:Simulated-Tempering-empty}
Let $\alpha \in [0,1)$ be any fixed constant.
Then the random graph $\GG(n,\frac{1}{2},k = \floor{n^\alpha})$ with a planted clique satisfies the following with probability $1-o(1)$ as $n\to \infty$. 

Consider the general ST dynamics given in \cref{sec:st_dfn} for monotone $1$-Lipschitz $\ham$ with $h_0 = 0$, arbitrary inverse temperatures $\beta_0<\beta_1<\ldots<\beta_m $ with $m = o(\log n)$ and $\beta_m = O(1)$, and arbitrary estimates $\hat{Z}(\beta_0) < \hat{Z}(\beta_1) < \cdots < \hat{Z}(\beta_m)$. 
 For any constants $c \in \N$, $\eps \in (0,1-\alpha)$, and $\gamma \in (0,1-\alpha]$, the ST dynamics starting from $(\emptyset, 0)$ will not reach within $n^c$ steps a pair of temperature and clique where
\begin{itemize}
    \item either the clique is of size at least $(1+\eps)\log n$, 
    \item or the clique is of intersection at least $\gamma \log n$ with the planted clique,
\end{itemize}
with probability $1-n^{-\omega(1)}$.
\end{theorem} 



In what follows we condition on both $\pupp(0.1)$ and $\pexp(\eta)$ for $\eta = \eps/2$, which by \cref{lem:rg_property1,lem:satisfy-cond} hold with probability $1-o(1)$ as $n\to \infty$.
Also, we may assume that
\begin{equation}\label{eq:BS}
\frac{\widehat{Z}(\beta_i)}{\widehat{Z}(\beta_{i+1})} \ge \frac{1}{n^{\sqrt{\frac{\log n}{m}}}}
\end{equation}
for all $i = 0,1,\dots,m-1$. 
Otherwise, for any clique $C$ of size $O(\log n)$ one has
\[
P_{\simt}((C,i), (C, i+1)) 
= \frac{1-\pa}{2} \min\left\{ 1, \frac{\widehat{Z}(\beta_i)}{\widehat{Z}(\beta_{i+1})} \exp\left[ \left( \beta_{i+1} - \beta_i \right) \ham_{|C|} \right] \right\}
\le
\frac{e^{O(\log n)}}{n^{\sqrt{\frac{\log n}{m}}}}
= \frac{1}{n^{\omega(1)}}
\]
since $\beta_m = O(1)$, $h_{|C|} \le |C| = O(\log n)$ and $m=o(\log n)$.
This means that the chain with high probability will never make a move to the inverse temperature $\beta_{i+1}$ in polynomially many steps, unless already having clique size, say, $\ge 10 \log n$. 
Since we are studying reaching cliques of size $(1+\eps)\log n$, we may assume that \cref{eq:BS} holds for all $i$, by removing those temperatures violating \cref{eq:BS} and those larger since the chain does not reach them with high probability in $\poly(n)$ steps.  
An immediate corollary of \cref{eq:BS} is that
\begin{equation}\label{eq:BS2}
\frac{\widehat{Z}(\beta_m)}{\widehat{Z}(\beta_{0})} 
\le n^{m\sqrt{\frac{\log n}{m}}}
= n^{\sqrt{m\log n}}
= n^{o(\log n)}.
\end{equation}

Let $q = (1+\eps) \log n$ and $s = \gamma \log n$. By the union bound we have
\begin{align}
&\Pr\Big( \exists t \in \N \wedge t \le T:\, (X_t,I_t) \in (\XX_{q,<s} \cup \XX_{<q,s}) \times [m] \;\Big\vert\; (X_0,I_0) = (\emptyset,0) \Big) \notag\\
\le{}&
T(m+1)n^4 \,\max_{t \in [T]} \,\max_{\ell \in [m]} \,\max_{r \in [s]} \max_{\sigma \in \XX_{q,r}} \E[W_{q,r}] \Pr\left( (X_t,I_t) = (\sigma,\ell) \mid (X_0,I_0) = (\emptyset,0) \right) \nonumber\\
&+ T(m+1)n^4 \,\max_{t \in [T]} \,\max_{\ell \in [m]} \,\max_{p \in [q]} \max_{\sigma \in \XX_{p,s}} \E[W_{p,s}] \Pr\left( (X_t,I_t) = (\sigma,\ell) \mid (X_0,I_0) = (\emptyset,0) \right)
\label{eq:goal-st-sfe1}
\end{align}
We will show that for all integer $t \ge 1$, all integer $\ell \in [m]$, all integer $r \le s$, and all clique $\sigma \in \XX_{q,r}$, it holds
\begin{equation}\label{eq:goal-st-sfe2}
\E[W_{q,r}] \Pr\left( (X_t,I_t) = (\sigma,\ell) \mid (X_0,I_0) = (\emptyset,0) \right) \le n^{-\omega(1)}, 
\end{equation}
and all integer $p \le q$, and all clique $\sigma \in \XX_{p,s}$, it holds
\begin{equation}\label{eq:goal-st-sfe3}
\E[W_{p,s}] \Pr\left( (X_t,I_t) = (\sigma,\ell) \mid (X_0,I_0) = (\emptyset,0) \right) \le n^{-\omega(1)}, 
\end{equation} 
The theorem then follows from \cref{eq:goal-st-sfe1,eq:goal-st-sfe2,eq:goal-st-sfe3}.

It will be helpful to consider the time-reversed dynamics and try to bound the probability of reaching $\emptyset$ when starting from a large clique $\sigma$. By reversibility, we have
\begin{equation}\label{eq:reversible-st-eq}
\Pr\left( (X_t,I_t) = (\sigma,\ell) \mid (X_0,I_0) = (\emptyset,0) \right)
= { \frac{\widehat{Z}(\beta_0)}{\widehat{Z}(\beta_\ell)} \frac{\exp\left( \beta_\ell \ham_q \right)}{\exp\left( \beta_0 \ham_0 \right)} } \Pr\left( (X_t,I_t) = (\emptyset,0) \mid (X_0,I_0) = (\sigma,\ell) \right).
\end{equation}
which is an application of \cref{fact:reversibility}.

Let $\eta \in (0,1)$ be a constant. 
Introduce a random walk $\{(Y_t,J_t)\}$ on $[n] \times [m]$ with transition matrix $P$ given by
\begin{align*}
P\left( (\size,j), (\size-1,j) \right) &= \frac{\pa \size}{n} \min \left\{ 1, \exp\left[ \beta_j \left( \ham_{\size-1} - \ham_\size \right) \right] \right\};\\
P\left( (\size,j), (\size+1,j) \right) &= \begin{cases}
\dfrac{\pa}{20 \cdot 2^\size} \min \left\{ 1, \exp\left[ \beta_j \left( \ham_{\size+1} - \ham_\size \right) \right] \right\}, & \text{for~} 0\le \size < \floor{(1-\eta) \log n}; \vspace{0.2em}\\
0, & \text{for~} \floor{(1-\eta) \log n} \le \size \le n;
\end{cases}\\
P\left( (\size,j), (\size,j') \right) &= \frac{1-\pa}{2} \min \left\{ 1, \frac{\widehat{Z}(\beta_j)}{\widehat{Z}(\beta_{j'})} \exp\left[ \left( \beta_{j'}-\beta_j \right) \ham_\size \right] \right\} \quad \text{for $j' = j \pm 1$}.
\end{align*}
and for all $(\size,j)$,
\[
P\left( (\size,j),(\size,j) \right) = 1 - \sum_{\size'=\size\pm 1} P\left( (\size,j),(\size',j) \right) - \sum_{j'=j \pm 1} P\left( (\size,j),(\size,j') \right).
\]
To be more precise, the above definition of $P\left( (\size,j),(\size',j') \right)$ applies when $(\size',j') \in [n] \times [m]$ and we assume $P\left( (\size,j),(\size',j') \right) = 0$ if $(\size',j') \notin [n] \times [m]$, e.g., when $j = 0$ and $j' = -1$.

We now calculate the stationary distribution of $P$ on states $(s,i)$ when restricted to $ \size  \leq \floor{(1-\eta) \log n}$. We start with proving that the random walk is time-reversible. Note that the random walk introduced is clearly aperiodic, positive recurrent and irreducible. Hence, by the Kolmogorov's criterion the random walk is time-reversible if and only if for any cycle in the state space, the probability the random walk moves along the cycle in one direction equals to the probability of moving in the opposite direction. Given that the minimal cycles in the finite box $[n] \times [m]$ are simply squares of the form $\{s, s+1 \} \times \{j, j+1\}$, it suffices to show that  for any $s \in [n-1], j \in [m-1],$ the criterion solely for these cycles, that is to show
\begin{align*}
  &P\left( (\size,j),(\size+1,j) \right)P\left( (\size+1,j),(\size+1,j+1) \right)P\left( (\size+1,j+1),(\size+1,j) \right)P\left( (\size+1,j),(\size,j) \right)  \\
  =&P\left( (\size,j),(\size,j+1) \right)P\left( (\size,j+1),(\size+1,j+1) \right)P\left( (\size+1,j+1),(\size,j+1) \right)P\left( (\size,j+1),(\size,j) \right),
\end{align*}which can be straightforwardly checked to be true.

We now calculate the stationary distribution. Using the reversibility and that $h_{\ell}$ is monotonically increasing in $\ell \in \mathbb{Z}$, we have for arbitrary $s,j$

\begin{align*}
\nu((s,j)) &= \nu((0,j)) \prod_{t = 1}^{s} \frac{P\left( (t,j), (t-1,j) \right)}{P\left( (t-1,j), (t,j) \right)} \\
&= \nu((0,j)) \prod_{t =1}^{s} \frac{ \frac{a t}{n} \min \left\{ \exp\left[ \beta_j \left( \ham_{t-1} - \ham_t \right) \right], 1 \right\} }{ \frac{a}{20 \cdot 2^{t}} \min \left\{ \exp\left[ \beta_j \left( \ham_t - \ham_{t-1} \right) \right], 1 \right\} } \\
&= \nu((0,j)) \prod_{t=1}^{s} \frac{20t \cdot 2^{t}}{n} \exp\left[ \beta_j \left( \ham_{t-1} - \ham_t \right) \right] \\
&= \nu((0,j)) \frac{20^{s} \cdot s!  \cdot 2^{\binom{s}{2}}}{n^{s}} \exp\left[ \beta_j (\ham_0-  \ham_s) \right].
\end{align*}Furthermore, since $h_0=0,$ we have again by reversibility,

\begin{align*}
\nu((0,j)) &= \nu((0,1)) \prod_{t = 1}^{j} \frac{P\left( (0,t), (0,t-1) \right)}{P\left( (0,t-1), (0,t) \right)} \\
&=  \nu((0,1))\prod_{t =1}^{j} \frac{ \frac{1-\pa}{2} \min \left\{ 1, \frac{\widehat{Z}(\beta_{t})}{\widehat{Z}(\beta_{t-1})} \exp\left[ \left( \beta_{t-1}-\beta_{j} \right) \ham_0 \right] \right\} }{ \frac{1-\pa}{2} \min \left\{ 1, \frac{\widehat{Z}(\beta_{t-1})}{\widehat{Z}(\beta_{t})} \exp\left[ \left( \beta_{t}-\beta_{t-1} \right) \ham_0 \right] \right\} } \\
&= \nu((0,1)) \prod_{t =1}^{j}   \frac{\widehat{Z}(\beta_{t})}{\widehat{Z}(\beta_{t-1})}   \\
& \propto  \widehat{Z}(\beta_{j}).
\end{align*} Combining the above, we conclude\begin{align}\label{eq:stat_distr_temp}
\nu((s,j)) & \propto  \widehat{Z}(\beta_{j})\frac{20^{s} \cdot s!  \cdot 2^{\binom{s}{2}}}{n^{s}} \exp\left[ -\beta_j  \ham_s \right].
\end{align}

The following lemma shows that $(Y_t,J_t)$ is stochastically dominated by the pair $(|X_t|,I_t)$.

\begin{lemma}\label{lem:stochastic-dominance-2d}
Let $\{(X_t,I_t)\}$ denote the Simulated Tempering process starting from some $X_0 = \sigma \in \XX_{q,*}$ and $I_0 = \ell$.
Let $\{(Y_t,J_t)\}$ denote the stochastic process described above with parameter $\eta \in (0,1)$ starting from $Y_0 = q$ and $J_0 = \ell$.
Assume that $G$ satisfies the conclusion of \cref{lem:satisfy-cond} with parameter $\eta$, then there exists a coupling $\{((X_t,I_t),(Y_t,J_t))\}$ of the two processes such that for all integer $t \ge 1$ it holds
\[
Y_t \le |X_t|
\quad\text{and}\quad
J_t \le I_t.
\]
In particular, for all integer $t \ge 1$ it holds
\begin{align*}
\Pr\left( (X_t,I_t) = (\emptyset,0) \mid (X_0,I_0) = (\sigma,\ell) \right) 
\le \Pr\left( (Y_t,J_t) = (0,0) \mid (Y_0,J_0) = (q,\ell) \right).  
\end{align*} 
\end{lemma}
\begin{proof}
We couple $\{(X_t,I_t)\}$ and $\{(Y_t,J_t)\}$ as follows. 
Suppose that $\{(Y_{t-1},J_{t-1})\} \le \{(X_{t-1},I_{t-1})\}$ for some integer $t \ge 1$.
We will construct a coupling of $\{(X_t,I_t)\}$ and $\{(Y_t,J_t)\}$ such that $\{(Y_t,J_t)\} \le \{(X_t,I_t)\}$.
With probability $\pa$, the two chains both attempt to update the first coordinate, and with probability $1-\pa$ the second. 

Consider first updating the first coordinate. Since the probability that $Y_t = Y_{t-1} + 1$ is less than $1/2$ and so does the probability of $|X_t| = |X_{t-1}| - 1$, we may couple $X_t$ and $Y_t$ such that $|X_t| - Y_t$ decreases at most one; namely, it never happens that $Y_t$ increases by $1$ while $X_t$ decreases in size. 
Thus, it suffices to consider the extremal case when $|X_{t-1}| = Y_{t-1} = \size$. Since $i = I_{t-1} \ge J_{t-1} = j$, we have $\beta_i \ge \beta_j$ and thus
\begin{align*}
\Pr\left( |X_t| = \size - 1 \mid |X_{t-1}| = \size, I_{t-1} = i \right)
&= \frac{\pa \size}{n} \min\left\{ 1, \exp\left[ \beta_i \left( \ham_{\size-1} - \ham_{\size} \right) \right] \right\} \\
&\le \frac{\pa \size}{n} \min\left\{ 1, \exp\left[ \beta_j \left( \ham_{\size-1} - \ham_\size \right) \right] \right\}
= P\left( (\size,j), (\size-1,j) \right)
\end{align*}
So we can couple $\{(X_t,I_t)\}$ and $\{(Y_t,J_t)\}$ such that either $|X_t| = Y_t$ or $|X_t| = \size$, $Y_t = \size-1$.
Meanwhile, recall that $A(X_{t-1})$ is the set of vertices $v$ such that $X_{t-1} \cup \{v\} \in \XX$. Then we have 
\[
|A(X_{t-1})| \ge \frac{n}{20\cdot 2^\size}
\]
whenever $\size \le n_\eta$ by \cref{lem:satisfy-cond}. Hence, we deduce that
\begin{align*}
\Pr\left( |X_t| = \size + 1 \mid |X_{t-1}| = \size, I_{t-1} = i \right)
&= \dfrac{\pa |\mathrm{ext}(X_{t-1})|}{n} \min \left\{ 1, \exp\left[ \beta_i \left( \ham_{\size+1} - \ham_\size \right) \right] \right\} \\
&\ge \dfrac{\pa \one\{\size \le n_\eta\}}{20 \cdot 2^\size} \min \left\{ 1, \exp\left[ \beta_j \left( \ham_{\size+1} - \ham_\size \right) \right] \right\}
= P\left( (\size,j), (\size+1,j) \right)
\end{align*}
So we can couple $\{(X_t,I_t)\}$ and $\{(Y_t,J_t)\}$ such that either $|X_t| = Y_t$ or $|X_t| = \size+1$ and $Y_t = \size$, as desired.

Next we consider update the second coordinate. Since the probability that $J_t = J_{t-1} + 1$ is less than $1/2$ and so does the probability of $I_t = I_{t-1} - 1$, we can couple $I_t$ and $J_t$ such that it never happens both $J_t = J_{t-1} + 1$ and $I_t = I_{t-1} - 1$. 
This means that it suffices for us to consider the extremal case where $I_t = J_t = i$.
Since $g = |X_{t-1}| \ge Y_{t-1} = \size$, we have $\ham_g \ge \ham_\size$ and therefore
\begin{align*}
\Pr\left( I_t = i-1 \mid I_{t-1} = i, |X_{t-1}| = g \right)
&= \frac{1-\pa}{2} \min \left\{ 1, \frac{\widehat{Z}(\beta_i)}{\widehat{Z}(\beta_{i-1})} \exp\left[ \left( \beta_{i-1}-\beta_i \right) \ham_g \right] \right\} \\
&\le \frac{1-\pa}{2} \min \left\{ 1, \frac{\widehat{Z}(\beta_i)}{\widehat{Z}(\beta_{i-1})} \exp\left[ \left( \beta_{i-1}-\beta_i \right) \ham_\size \right] \right\}
= P\left( (\size,i), (\size,i-1) \right).
\end{align*}
and similarly,
\begin{align*}
\Pr\left( I_t = i+1 \mid I_{t-1} = i, |X_{t-1}| = g \right)
&= \frac{1-\pa}{2} \min \left\{ 1, \frac{\widehat{Z}(\beta_i)}{\widehat{Z}(\beta_{i+1})} \exp\left[ \left( \beta_{i+1}-\beta_i \right) \ham_g \right] \right\} \\
&\ge \frac{1-\pa}{2} \min \left\{ 1, \frac{\widehat{Z}(\beta_i)}{\widehat{Z}(\beta_{i+1})} \exp\left[ \left( \beta_{i+1}-\beta_i \right) \ham_\size \right] \right\}
= P\left( (\size,i), (\size,i+1) \right).
\end{align*}
Therefore, we can couple $I_t$ and $J_t$ such that only one of the followings can happen:
\begin{enumerate}[(i)]
\item $I_t = J_t$;
\item $I_t = i$ and $J_t = i-1$;
\item $I_t = i+1$ and $J_t = i$.
\end{enumerate}
Therefore, one always has $I_t \ge J_t$ under the constructed coupling.
\end{proof}

The next lemma upper bounds the $t$-step transition probability $\Pr\left(Y_t = \qzero \mid Y_0 = q\right)$.

\begin{lemma}\label{lem:YJ-ub-prob}
Let $\{(Y_t,J_t)\}$ denote the Markov process on $\Z^2$ described above with parameter $\eta \in (0,1)$ starting from $Y_0 = q = \rho \log n$ and $J_0 = j$.
Then for all integer $t \ge 1$ we have
\begin{equation*}
\Pr\left( (Y_t,J_t) = (0,0) \mid (Y_0,J_0) = (q,\ell) \right) 
\le \exp\left[ (\ln2) (\log n)^2 \left( - \rho' + \frac{(\rho')^2}{2} + o(1) \right) \right] 
{ \frac{\widehat{Z}(\beta_\ell)}{\widehat{Z}(\beta_0)} \frac{\exp\left( \beta_0 \ham_0 \right)}{\exp\left( \beta_\ell \ham_q \right)} }.
\end{equation*}
where $\rho' = \min\{ \rho, 1-\eta \}$.
\end{lemma}

\begin{proof}
When $\rho > 1-\eta$, namely $q > (1-\eta) \log n$, the chain will first move to $(1-\eta) \log n$ before reaching $p$. 
By \cref{fact:reversibility}, we have
\[
\Pr\left((Y_t,J_t) = (0,0) \mid (Y_0,J_0) = (q,\ell) \right) 
= P^t((q,\ell), (0,0)) 
= \frac{\nu((q,\ell))}{\nu((0,0))} P^t((0,0),(q,\ell)) 
\le \frac{\nu((q,\ell))}{\nu((0,0))}.
\] By \eqref{eq:stat_distr_temp},
\begin{align*}
\frac{\nu((q,\ell))}{\nu((0,0))}.
&= \frac{\widehat{Z}(\beta_{\ell})}{ \widehat{Z}(\beta_0)}  \frac{20^{q} \cdot (q!) \cdot 2^{\binom{q}{2}}}{n^{q}} \exp\left[ \beta_0 \ham_0 - \beta_\ell \ham_q  \right]   \\
&= \frac{\widehat{Z}(\beta_\ell)}{\widehat{Z}(\beta_0)} \frac{\exp\left( \beta_0 \ham_0 \right)}{\exp\left( \beta_\ell \ham_q \right)} \exp\left[ (\ln2) (\log n)^2 \left( - \rho + \frac{\rho^2}{2} + o(1) \right) \right].
\end{align*}
For $\rho > 1-\eta$, let $\tau$ be the first time that the chain reach size $q'=(1-\eta)\log n$. Then we have
\begin{align*}
&\Pr\left( (Y_t,J_t) = (0,0) \mid (Y_0,J_0) = (q,\ell) \right) \\
={}& \sum_{t'=0}^t
\Pr(\tau = t') \sum_{\ell' \in [m]} 
\Pr\left( (Y_{t'},J_{t'}) = (q', \ell') \mid (Y_0,J_0) = (q,\ell), \tau = t'\right) \\
&\cdot \Pr\left( (Y_{t-t'},J_{t-t'}) = (0,0) \mid (Y_0,J_0) = (q', \ell') \right)\\
&\le \max_{\ell'\in[m]}\frac{\nu((q',\ell'))}{\nu((0,0))}\\
&\le \max_{\ell'\in[m]} \frac{\widehat{Z}(\beta_{\ell'})}{\widehat{Z}(\beta_0)} \frac{\exp\left( \beta_0 \ham_0 \right)}{\exp\left( \beta_{\ell'} \ham_{q'} \right)} \exp\left[ (\ln2) (\log n)^2 \left( - \rho' + \frac{(\rho')^2}{2} + o(1) \right) \right].
\end{align*}
The lemma then follows from the fact that for all $\ell' \in [m]$
\[
\frac{\widehat{Z}(\beta_{\ell'})}{\widehat{Z}(\beta_\ell)} \frac{\exp\left( \beta_\ell \ham_q \right)}{\exp\left( \beta_{\ell'} \ham_{q'} \right)}
\le n^{o(\log n)} e^{O(\log n)}
= n^{o(\log n)},
\]
where we use \cref{eq:BS2}.
\end{proof}

We now present the proof of \cref{thm:Simulated-Tempering-empty} provided \cref{lem:stochastic-dominance-2d,lem:YJ-ub-prob}.

\begin{proof}[Proof of \cref{thm:Simulated-Tempering-empty}]
Recall that $\gamma < 2(1-\alpha)$. 
As will be clear later, we define
\[
\eta = \sqrt{(1-\alpha)\gamma - \frac{\gamma^2}{2}}. 
\]
We assume that our graph satisfies the conclusions of \cref{basic_rg,lem:satisfy-cond} with parameter $\eta$.

Consider first \cref{eq:goal-st-sfe2}. From \cref{lem:stochastic-dominance-2d,lem:YJ-ub-prob}, we deduce that
\begin{align*}
&~~~~ \E[W_{q,r}] \Pr\left( (X_t,I_t) = (\sigma,\ell) \mid (X_0,I_0) = (\emptyset,0) \right) \\
&= \E[W_{q,r}] \frac{\widehat{Z}(\beta_0)}{\widehat{Z}(\beta_\ell)} \frac{\exp\left( \beta_\ell \ham_q \right)}{\exp\left( \beta_0 \ham_0 \right)} \Pr\left( (X_t,I_t) = (\emptyset,0) \mid (X_0,I_0) = (\sigma,\ell) \right)  \\
&\le \E[W_{q,0}] \frac{\widehat{Z}(\beta_0)}{\widehat{Z}(\beta_\ell)} \frac{\exp\left( \beta_\ell \ham_q \right)}{\exp\left( \beta_0 \ham_0 \right)} \Pr\left( (Y_t,J_t) = (0,0) \mid (Y_0,J_0) = (q,\ell) \right) \\
&\le \exp\left[ (\ln2) (\log n)^2 \left( \rho - \frac{\rho^2}{2} + o(1) \right) \right] \frac{\widehat{Z}(\beta_0)}{\widehat{Z}(\beta_\ell)} \frac{\exp\left( \beta_\ell \ham_q \right)}{\exp\left( \beta_0 \ham_0 \right)} \\
&~~~\cdot \exp\left[ (\ln2) (\log n)^2 \left( - \rho' + \frac{(\rho')^2}{2} + o(1) \right) \right] \frac{\widehat{Z}(\beta_\ell)}{\widehat{Z}(\beta_0)} \frac{\exp\left( \beta_0 \ham_0 \right)}{\exp\left( \beta_\ell \ham_q \right)} \\
&= \exp\left[ (\ln2) (\log n)^2 \left( -\frac{\eps^2}{2} + \frac{\eta^2}{2} + o(1) \right) \right]
\le \exp\left[ (\ln2) (\log n)^2 \left( -\frac{3}{8}\eps^2 + o(1) \right) \right],
\end{align*}
where recall that $\rho = 1+\eps$ and $\rho' = 1-\eta$ and we choose $\eta = \eps/2$. 

For \cref{eq:goal-st-sfe3}, by similar argument we have
\begin{align*}
&~~~~ \E[W_{p,s}] \Pr\left( (X_t,I_t) = (\sigma,\ell) \mid (X_0,I_0) = (\emptyset,0) \right) \\
&\le \exp\left[ (\ln2) (\log n)^2 \left( \frac{1}{2} - \frac{1}{2}(1-\alpha)^2 + o(1) \right) \right] \frac{\widehat{Z}(\beta_0)}{\widehat{Z}(\beta_\ell)} \frac{\exp\left( \beta_\ell \ham_q \right)}{\exp\left( \beta_0 \ham_0 \right)} \\
&~~~\cdot \exp\left[ (\ln2) (\log n)^2 \left( - \rho' + \frac{(\rho')^2}{2} + o(1) \right) \right] \frac{\widehat{Z}(\beta_\ell)}{\widehat{Z}(\beta_0)} \frac{\exp\left( \beta_0 \ham_0 \right)}{\exp\left( \beta_\ell \ham_q \right)} \\
&= \exp\left[ (\ln2) (\log n)^2 \left( - \frac{1}{2}(1-\alpha)^2 + \frac{\eta^2}{2} + o(1) \right) \right]
\le \exp\left[ (\ln2) (\log n)^2 \left( -\frac{3}{8}(1-\alpha)^2 + o(1) \right) \right],
\end{align*}
where recall that $\rho = 1$ (for maximizing $\E[W_{p,s}]$) and $\rho' = 1-\eta$ and we choose $\eta = \eps/2 \le (1-\alpha)/2$. 

Therefore, we obtain \cref{eq:goal-st-sfe2,eq:goal-st-sfe3}. The theorem then follows.
\end{proof}

\section{Conclusion}

In this work we revisit the work by Jerrum \cite{Jer92} that large cliques elude the Metropolis process. We extend \cite{Jer92} by establishing the failure of the Metropolis process (1) for all planted clique sizes $k=n^{\alpha}$ for any constant $\alpha \in (0,1),$ (2) for arbitrary temperature and Hamiltonian vector (under worst-case initialization), (3) for a large family of of temperatures and Hamiltonian vectors (under the empty clique initialization) and obtain as well (4) lower bounds for the performance of the simulated tempering dynamics.

An important future direction would be to explore the generality of our proposed reversibility and birth and death process arguments which allowed us to prove the failure of the Metropolis process when initialized at the empty clique. It is interesting to see whether the proposed method can lead to MCMC lower bounds from a specific state in other inference settings beyond the planted clique model.

Moreover, it would be interesting to see if our results can be strengthened even more. First, a current shortcoming of our lower bounds for the Metropolis process when initialized from the empty clique do not completely cover the case where $\beta=C \log n$ for an arbitrary constant $C>0$. While we almost certainly think the result continues to hold in this case, some new idea seem to be needed to prove it. Second, it seems interesting to study the regime where $\alpha=1-o(1).$ Recall that there are polynomial-time algorithms that can find a clique of size $(\log n/\log\log n)^2$ whenever a (worst-case) graph has a clique of size $n/(\log n)^b$, for some constant $b>0$ \cite{feige_max_clique}. If our lower bounds for the Metropolis process could be extended to the case $\alpha=1-O(\log \log n/\log n)$, this would mean that for some $k$ the Metropolis process fails to find in polynomial-time a clique of size $(1+\epsilon)\log n$ when a $k$-clique is planted in $\GG(n,1/2),$ while some other polynomial-time algorithm can find a clique of size $(1+\epsilon)\log n$ on every (worst-case) graph which has a clique of size $k.$ Such a result, if true, will make the failure of the Metropolis process even more profound.

\section*{Acknowledgement}
EM and IZ are supported by the Simons-NSF grant DMS-2031883 on the Theoretical Foundations of Deep Learning and the Vannevar Bush Faculty Fellowship ONR-N00014-20-1-2826. EM is also supported by the Simons Investigator award (622132).


\newpage

\bibliographystyle{alpha}
\bibliography{all,mcmc}

\newpage
\appendix

\section{Deferred Proofs}
\begin{proof}[Proof of \cref{basic_rg}]

For part (1), notice that by linearity of expectation and the elementary application of Stirling's formula that for $m_2 \leq m_1$ with $ m_2=o(m_1)$ it holds $\binom{m_1}{m_2}=(m_1/m_2)^{m_2(1+o(1))}$ we have
\begin{align*}
     \E[W_{q,r}]&=\binom{k}{r}\binom{n-k}{q-r}2^{\binom{r}{2}-\binom{q}{2}}\\
     &=\binom{n^{\alpha}}{\floor{\gamma \log n}}\binom{n-n^{\alpha}}{\floor{\rho \log n}-\floor{\gamma \log n}}2^{\binom{\floor{\gamma \log n}}{2}-\binom{\floor{\rho \log n}}{2}}\\
     &=\exp\left[ (\ln2) (\log n)^2 \left( \alpha \gamma +(\rho-\gamma)+\frac{\gamma^2}{2}-\frac{\rho^2}{2} + o(1) \right) \right]\\
     &=\exp\left[ (\ln2) (\log n)^2 \left( \rho - \frac{\rho^2}{2} - (1-\alpha)\gamma + \frac{\gamma^2}{2} + o(1) \right) \right].
\end{align*}

For part (2), notice that $W_{q,0}$ is distributed as the number of $q$-cliques in $\GG(n-k,\frac{1}{2})$. Hence, standard calculation (e.g. \cite[Proof of Theorem 1]{bollobas_erdos_1976}) prove that since $\rho=2-\Omega(1),$
$$\frac{\mathrm{Var}(W_{q,0})}{\E[W_{q,0}]^2}=O(\frac{q^4}{n^2})=O(\frac{1}{n}).$$Hence, Chebyshev's inequality yields that with probability $1-O(1/n),$ $W_{q,0} \geq \frac{1}{2}\E[W_{q,0}].$ Taking a union bound over the different values of $q=O(\log n)$ completes the proof of this part.

Finally, part (3) follows directly from part (1), Markov's inequality and a union bound over the possible values of $r,q$.
\end{proof}
\begin{proof}[Proof of \cref{lem:satisfy-cond}] It clearly suffices to establish this result for $k=0$, i.e. for an the random graph $G(n,\frac{1}{2}).$
For any fixed $|U| \le (1-\eta) \log n,$ $|A(U)|$ follows a Binomial distribution with $n-|U|$ trials and probability $\frac{1}{2^{|U|}}.$ In particular, it has a mean $(1+o(1))\frac{n}{2^{|U|}}=\Omega(n^{\eta})$. Hence, by Hoeffding's inequality with probability $1-\exp(-\Omega(n^{\eta}))$ it holds  $|A(U)| \ge \frac{n}{20 \cdot 2^{|U|}}$. As there are only $\binom{n}{\floor{\log n}}=n^{O(\log n)}$ the result follows from a union bound over $|U|$.
\end{proof}
\end{document}